\documentclass[1p]{elsarticle} 

\usepackage{hyperref}
\usepackage[table]{xcolor}
\usepackage{amsmath,amsfonts,amssymb}
\usepackage[inline]{enumitem}
\usepackage{latexsym,dsfont,marvosym}
\usepackage{wrapfig}
\usepackage{algorithmic,algorithm}%
\usepackage{multirow}
\usepackage{xifthen}
\usepackage{tikz}
\usetikzlibrary{calc} 
\usetikzlibrary{decorations.pathreplacing}

  \tikzstyle{labels} = [
    label/.style={draw,circle, thin, inner sep=0.6mm},
    initial/.style={label, very thick},
    cover/.style={label, fill=gray!30},
  ]
  \tikzset{every loop/.style={max distance=8mm}}

\DeclareMathOperator*{\argmax}{\arg\!\max}
\providecommand*{\ifempty}[3]{\ifthenelse{\isempty{#1}}{#2}{#3}}
\newcommand{\parensmathoper}[2]{\ensuremath{#1\ifempty{#2}{}{(#2)}}}

\newcommand{\cf}{\emph{cf.}}

\newcommand{\naturals}{\mathbb{N}}

\newcommand{\M}{\mathcal{M}}
\newcommand{\N}{\mathcal{N}}

\newcommand{\MC}[2][k]{\text{\rm MC}_{#2}(#1)}

\newcommand{\C}{\mathcal{C}}

\newcommand{\dist}[1][\lambda]{\delta_{#1}} 
\newcommand{\discr}[2][\lambda]{\gamma_{#1}^{#2}} 
\newcommand{\ddiscr}[2][\lambda]{\beta_{#1}^{#2}} 

\newcommand{\instance}[1]{\langle #1 \rangle}
\newcommand{\VC}{\textsc{Vertex Cover}}

\newcommand{\CBA}[1][\lambda]{\text{CBA-}{#1}}
\newcommand{\BA}[1][\lambda]{\text{BA-}{#1}}
\newcommand{\MSAB}[1][1]{\text{MSAB-}{#1}}
\newcommand{\SBA}[1][1]{\text{SBA-}{#1}}

\newcommand{\E}[2][]{\mathbf{E}[#2\ifempty{#1}{}{ |\, #1}]}

\newcommand{\D}{\mathcal{D}}

\renewcommand{\Pr}[1][]{\parensmathoper{\mathbb{P}_{#1}}}

\newcommand{\coupling}[2]{\Omega(#1, #2)}

\newcommand{\set}[2]{\left\{ #1 \ifempty{#2}{}{\mid #2} \right\}}

\newcommand{\norm}[2][\infty]{\| #2 \|_{#1}}
\newcommand{\K}[1][d]{\parensmathoper{\parensmathoper{\mathcal{K}}{#1}}}

 \newtheorem{theorem}{Theorem} 
 \newtheorem{lemma}[theorem]{Lemma} 
 \newtheorem{corollary}[theorem]{Corollary} 
 
 \newproof{proof}{Proof}
 
 \newdefinition{definition}{Definition}
 \newdefinition{remark}{Remark}
 \newdefinition{example}{Example}

\journal{Journal of Logical and Algebraic Methods in Programming}

\begin{document}

\begin{frontmatter}

\title{On the Metric-based Approximate Minimization \\ of Markov Chains\tnoteref{mytitlenote}\tnoteref{exnote}}
\tnotetext[exnote]{This article in an extended version of a paper accepted for publication at ICALP 2017~\cite{BacciLM:icalp17}. The current version provides proofs omitted in the original paper, additional examples, and a revised section on experimental results. Remarkably, we revise the proof of~\cite[Theorem 14]{BacciLM:icalp17} which contains a flaw.}
\tnotetext[mytitlenote]{Work supported by the EU 7th Framework Programme (FP7/2007-13)
under Grants Agreement nr.318490 (SENSATION), nr.601148 (CASSTING), the 
Sino-Danish Basic Research Center IDEA4CPS funded by Danish National Research 
Foundation and National Science Foundation China, the ASAP Project (4181-00360) funded by the Danish Council for Independent Research, the ERC Advanced Grant LASSO, and the Innovation Fund Denmark center DiCyPS.}


\author[aau]{Giovanni Bacci\corref{cor}}
\ead{giovbacci@cs.aau.dk}
\author[aau]{Giorgio Bacci}
\ead{grbacci@cs.aau.dk}
\author[aau]{Kim G. Larsen}
\ead{kgl@cs.aau.dk}
\author[aau]{Radu Mardare}
\ead{mardare@cs.aau.dk}

\address[aau]{Department of Computer Science, Aalborg University, Denmark}
\cortext[cor]{Corresponding author}

%
%
%

\begin{abstract}
In this paper we address the approximate minimization problem of Markov Chains (MCs) from a behavioral metric-based perspective. Specifically, given a finite MC and a positive integer $k$, we are looking for an MC with at most $k$ states having \emph{minimal} distance to the original. The metric considered in this work is the bisimilarity distance of Desharnais et al.. For this metric we show that (1) optimal approximations always exist; (2) the problem has a bilinear program characterization; and (3) prove that its threshold problem is in PSPACE and NP-hard. 

In addition to the bilinear program solution, we present an approach inspired by expectation maximization techniques for computing suboptimal solutions to the problem. 
Experiments suggest that our method gives a practical approach that outperforms the bilinear 
program implementation run on state-of-the-art bilinear solvers.
\end{abstract}


\end{frontmatter}


\section{Introduction}
\label{sec:intro}

Minimization of finite automata, i.e., the process of transforming a given finite automaton into an equivalent one with minimum number of states, has been a major subject since the 1950s 
due to its fundamental importance for any implementation of finite automata tools.

The first algorithm for the minimization of deterministic finite automata (DFAs) is due to Moore~\cite{Moore56},
with time complexity $O(n^2s)$, later improved by the now classical Hopcroft's algorithm~\cite{Hopcroft71} to $O(ns \log n)$, where $n$ is the number of states and $s$ the size of the alphabet.
Their algorithms are based on a partition refinement of the states into equivalence classes of the \emph{Myhill-Nerode equivalence relation}.
Partition refinement has been employed in the definition of efficient minimization procedures for a wide variety of automata: by Kanellakis and Smolka~\cite{KanellakisS83,KanellakisS90} for the minimization of labelled transition systems (LTSs) w.r.t.\ Milner's strong bisimulation~\cite{Milner80}; by Baier~\cite{Baier96} for the reduction of Markov Chains (MCs) w.r.t.\ Larsen and Skou's probabilistic bisimulation~\cite{LarsenS91}; by Alur et al.~\cite{AlurCHDW92} and by Yannakakis and Lee~\cite{Yannakakis97}, respectively, for the minimization of timed transition systems and timed-automata. This technique was used also in parallel and distributed implementations of 
the above algorithms~\cite{ZhangS92,Blom2005}, and in the online reachability analysis of transition systems~\cite{LeeY92}. 

In~\cite{JouS90}, Jou and Smolka observed that for reasoning about the behavior of probabilistic systems (and more in general, all type of quantitative systems), rather than equivalences, a notion of distance is more reasonable in practice, since it permits ``\textit{a shift in attention from equivalent processes to probabilistically similar processes}''. 
This observation motivated the development of metric-based semantics for quantitative systems,
that consists in proposing $1$-bounded pseudometrics capturing the similarities of the behaviors
in the presence of small variations of the quantitative data. These pseudometrics generalize behavioral equivalences in the sense that, two processes are at distance $0$ iff they are equivalent, and at distance $1$ if no significant similarities can be observed between them.

The first proposal of a behavioral pseudometric is due to Desharnais et al.~\cite{DesharnaisGJP99} 
on labelled MCs, a.k.a.\ \emph{probabilistic bisimilarity distance}, with the property that two MCs are at distance $0$ iff they are probabilistic bisimilar.
Its definition is parametric on a discount factor $\lambda \in (0,1]$ that controls the significance of the future steps in the measurement.
This pseudometric has been greatly studied by van Breugel and Worrell \cite{BreugelW:icalp01,BreugelW06} who noticed, among other notable results, its relation with the Kantorovich distance on probability distributions and provided a polynomial-time algorithm for its computation~\cite{ChenBW12}.

The introduction of metric-based semantics motivated the interest in the approximate minimization
of quantitative systems. The goal of approximate minimization 
is to start from a minimal automaton and produce a smaller automaton that is close to the given
one in a certain sense. The desired size of the approximating automaton is given as input. 
Inspired by the aggregation of equivalent states typical of partition refinement techniques, 
in~\cite{FernsPP04}, the approximate minimization problem has been approached by aggregating  
states having relative smaller distance. An example of this approach on MCs using the bisimilarity distance of Desharnais et al.\ is shown in Figure~\ref{fig:aggregation}.
\begin{figure}[h]
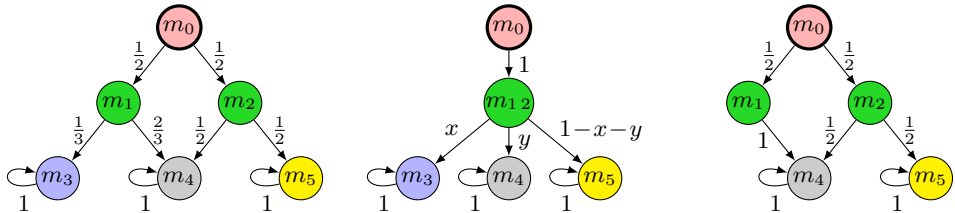

\def\h{0.8cm}
\def\v{1cm}
\begin{align*}
\tikz[labels, baseline={(current bounding box.center)},
label/.append style={inner sep=0.4mm}, every node/.append style={font=\small},
n1/.style={fill=red!30}, n2/.style={fill=blue!30}, n3/.style={fill=green!70!gray},
n4/.style={fill=yellow}, n5/.style={fill=gray!40}]{ 
  \draw (0,0) node[label, initial, n1] (m0) {$m_0$}; 
  \draw ($(m0)+(down:\v)+(left:\h)$) node[label, n3] (m1) {$m_1$};
  \draw ($(m0)+(down:\v)+(right:\h)$) node[label, n3] (m2) {$m_2$};
  \draw ($(m1)+(down:\v)+(left:\h)$) node[label, n2] (m3) {$m_3$};
  \draw ($(m2)+(down:\v)+(right:\h)$) node[label, n4] (m5) {$m_5$};
  \draw ($(m3)!0.5!(m5)$) node[label, n5] (m4) {$m_4$};
  \path[-latex, font=\small]
    (m0) edge node[left, near start] {$\frac{1}{2}$} (m1)
             edge node[right, near start] {$\frac{1}{2}$} (m2)
    (m1) edge node[left, near start] {$\frac{1}{3}$} (m3)
             edge node[right, near start] {$\frac{2}{3}$} (m4)
    (m2) edge node[left, near start] {$\frac{1}{2}$} (m4)
             edge node[right, near start] {$\frac{1}{2}$} (m5)
    (m3) edge[loop left] node[very near start, below] {$1$} (m3)
    (m4) edge[loop left] node[very near start, below] {$1$} (m4)
    (m5) edge[loop left] node[very near start, below] {$1$} (m5)
    ;
}
&&
\tikz[labels, baseline={(current bounding box.center)},
label/.append style={inner sep=0.4mm}, every node/.append style={font=\small},
n1/.style={fill=red!30}, n2/.style={fill=blue!30}, n3/.style={fill=green!70!gray},
n4/.style={fill=yellow},n5/.style={fill=gray!40}]{ 
  \draw (0,0) node[label, initial, n1] (m0) {$m_0$}; 
  \draw ($(m0)+(down:\v)+(left:0.5*\h)$) coordinate (m1);
  \draw ($(m0)+(down:\v)+(right:0.5*\h)$) coordinate (m2);
  \draw ($(m0)+(down:\v)$) node[label, n3, inner sep=0.1mm] (n2) {$m_{1\,2}$};
  \draw ($(m1)+(down:\v)+(left:\h)$) node[label, n2] (m3) {$m_3$};
  \draw ($(m2)+(down:\v)+(right:\h)$) node[label, n4] (m5) {$m_5$};
  \draw ($(m3)!0.5!(m5)$) node[label, n5] (m4) {$m_4$};
  \path[-latex, font=\small]
    (m0) edge node[right] {$1$} (n2)
    (n2) edge node[left, near start, xshift=-1mm] {$x$} (m3)
             edge node[right] {$y$} (m4)
             edge node[right, near start, xshift=1mm] {$1\!-\!x\!-\!y$} (m5)
    (m3) edge[loop left] node[very near start, below] {$1$} (m3)
    (m4) edge[loop left] node[very near start, below] {$1$} (m4)
    (m5) edge[loop left] node[very near start, below] {$1$} (m5)
    ;
}
&&
\tikz[labels, baseline={(current bounding box.center)},
label/.append style={inner sep=0.4mm}, every node/.append style={font=\small},
n1/.style={fill=red!30}, n2/.style={fill=blue!30}, n3/.style={fill=green!70!gray},
n4/.style={fill=yellow}, n5/.style={fill=gray!40}]{ 
  \draw (0,0) node[label, initial, n1] (m0) {$m_0$}; 
  \draw ($(m0)+(down:\v)+(left:\h)$) node[label, n3] (m1) {$m_1$};
  \draw ($(m0)+(down:\v)+(right:\h)$) node[label, n3] (m2) {$m_2$};
  \draw ($(m1)+(down:\v)+(left:\h)$) coordinate (m3);
  \draw ($(m2)+(down:\v)+(right:\h)$) node[label, n4] (m5) {$m_5$};
  \draw ($(m3)!0.5!(m5)$) node[label, n5] (m4) {$m_4$};
  \path[-latex, font=\small]
    (m0) edge node[left, near start] {$\frac{1}{2}$} (m1)
             edge node[right, near start] {$\frac{1}{2}$} (m2)
    (m1) edge node[left] {$1$} (m4)
    (m2) edge node[left, near start] {$\frac{1}{2}$} (m4)
             edge node[right, near start] {$\frac{1}{2}$} (m5)
    (m4) edge[loop left] node[very near start, below] {$1$} (m4)
    (m5) edge[loop left] node[very near start, below] {$1$} (m5)
    ;
}
\end{align*}
\caption{An MC $\M$ (left) with initial state $m_0$; generic $5$-states approximant of $\M$ aggregating $m_1$ and $m_2$ (for $x,y \in [0,1]$ such that $x+y \leq 1$) (middle); an optimal $5$-states approximant of $\M$ (right). Labels are represented by different colors.}
\label{fig:aggregation}
\end{figure}
Let $\M$ be the MC on the left and assume we want to approximate it by an MC 
with at most $5$ states. 
Since $m_1,m_2$ are the only two states at distance less than $1$, the most natural choice for an aggregation shall collapse $m_1$ and $m_2$, obtaining an instance $\N_{x,y}$ of the MC in the middle for some $x,y \geq 0$ such that $x + y \leq 1$. Any MC constructed in this way will not be closer than $\frac{1}{4}$ from $\M$. 
However, the MC on the right is a closer approximant of $\M$, at distance 
$\frac{1}{6}$ from it, showing that the approximate aggregation of states does not necessarily yield the closest optimal solution. 

In this paper we address the issue of finding \emph{optimal} solutions to the approximate minimization problem.
Specifically we aim to answer to the following problem, left open in~\cite{FernsPP04}:
``\textit{given a finite MC and a positive integer $k$, what is its `best' $k$-state approximant? Here by `best' we mean a $k$-state MC at minimal distance to the original}''.
We refer to this problem as \emph{Closest Bounded Approximant} (CBA) and we present the following results 
related to it.
\begin{enumerate}[topsep=0pt, noitemsep, fullwidth, itemindent=\parindent]
  \item We characterize CBA as a bilinear optimization problem, proving the  
  existence of \emph{optimal} solutions. As a consequence of this result, approximations of 
  optimal solutions can be obtained by checking the feasibility of 
  bilinear matrix inequalities (BMIs)~\cite{KocvaraS03,PENBMI}.
    
  \item We provide upper- and lower-bound complexity results for the threshold problem of CBA, 
  called \emph{Bounded Approximant} problem (BA), that asks whether there exists a $k$-state 
  approximant with distance from the original MC bounded by a given rational threshold.  
  We show that BA is in PSPACE and NP-hard. 
  
  \item We introduce the \emph{Minimum Significant Approximant Bound} (MSAB) problem, that asks
  what is the minimum size $k$ for an approximant to have some significant 
  similarity to the original MC (i.e., at distance strictly less than $1$). 
  We show that this problem is NP-complete when one considers the undiscounted bisimilarity distance.
  
  \item Finally, we present an algorithm for finding suboptimal solutions of CBA that is inspired by 
  Expectation Maximization (EM) techniques~\cite{McLachlanK08,BenediktLW13}. 
  Experiments suggest that our method gives a practical approach that outperforms the
  bilinear program implementation ---state-of-the-art bilinear solvers~\cite{PENBMI} 
  fail to handle MCs with more than 5 states!

\end{enumerate} 

\paragraph*{Related Work}
In~\cite{FranceschinisM94}, the approximate minimization of MCs is addressed via the 
notion of \emph{quasi-lumpability}. An MC is quasi-lumpable if the given aggregations of the states 
can be turned into actual bisimulation-classes 
by a small perturbation of the transition probabilities. This approach differs from ours 
since there is no relation to a proper notion of behavioral distance (the approximation 
is w.r.t.\ the supremum norm of the difference of the stochastic matrices) and we do not consider any 
approximate aggregation of states.
In~\cite{BallePP15}, Balle et al.\ consider the approximate minimization of weighted finite automata (WFAs).
Their method is via a truncation of a canonical normal form for WFAs that they introduced for the SVD 
decomposition of infinite Hankel matrices.
Both~\cite{FranceschinisM94} and \cite{BallePP15} do not consider the issue of finding the \emph{closest}
approximant, which is the main focus of this paper, instead they give upper bounds on the distance 
from the given model.

\paragraph*{Synopsis} Section~\ref{sec:MC} introduces some notation and definitions used in the paper, concluding with a useful characterization of the bisimilarity distance. Section~\ref{sec:OBA} formalizes the $\CBA$ problem and characterizes it as a bilinear optimization problem. Sections~\ref{sec:comeplexityAM} and~\ref{sec:MSAB} introduce two decision problems closely related with $\CBA$, namely $\BA$ and $\MSAB$, and study their computational complexity. Sections~\ref{sec:EM} and~\ref{sec:experiments} present two efficient heuristics for computing sub-optimal solutions for $\CBA$ and discuss about their performances and limitations. We conclude with Section~\ref{sec:conclusion} discussing possible applications of our results and interesting ideas for future work. 

\section{Markov Chains and Bisimilarity Pseudometric} \label{sec:MC}
In this section we introduce the notation and recall the definitions of (discrete-time) \emph{Markov chains} (MCs), \emph{probabilistic bisimilarity} of Larsen and Skou~\cite{LarsenS91}, and the \emph{probabilistic bisimilarity pseudometric} of Desharnais et al.~\cite{DesharnaisGJP04}.

For $R \subseteq X \times X$ an equivalence relation, $X /_R$ denotes its quotient set and $[x]_R$ denotes the $R$-equivalence class of $x \in X$. 

We denote by $\D(X)$ the set of discrete probability distributions on $X$, i.e., functions $\mu \colon X \to [0,1]$, such that $\mu(X) = 1$, where $\mu(E) = \sum_{x \in E} \mu(x)$ for $E \subseteq X$. For $x \in X$, the Dirac distribution concentrated at $x$ is the function $1_x \colon X \to [0,1]$ defined by $1_x(y) = 1$ if $x = y$, $0$ otherwise. 

For $a_1, \dots, a_n \in [0,1]$ such that $\sum_{i = 0}^n a_i = 1$, we denote by $\sum_{i = 1}^n a_i \cdot \mu_i $ the convex combination of $\mu_1, \dots, \mu_n \in \D(X)$, defined as
\begin{equation*}
\textstyle (\sum_{i = 1}^n a_i \cdot \mu_i)(x) = \sum_{i = 1}^n a_i \cdot \mu_i(x) \,.
\end{equation*}
The support of a probability distribution $\mu \in \D(X)$ is defined by 
\begin{equation*}
\mathit{support}(\mu) = \set{x \in X}{\mu(x) > 0} \,.
\end{equation*}
For example, given $x,y \in X$, $\mathit{support}(\frac{1}{3} \cdot 1_x + \frac{2}{3} \cdot 1_y) = \{x,y\}$.

In what follows we fix a countable set $L$ of labels. 

\begin{definition}[{Markov Chain}] \label{def:MC}
A \emph{Markov chain} is a tuple $\M = (M, \tau, \ell)$ consisting 
of a finite nonempty \emph{set of states} $M$, a \emph{transition distribution function}
$\tau \colon M \to \D(M)$, and a \emph{labelling function} $\ell \colon M \to L$.
\end{definition}
Intuitively, if $\M$ is in state $m$ it moves to state $m'$ with probability $\tau(m)(m')$. Labels represent atomic properties that hold in certain states. The set of labels of $\M$ is denoted by $L(\M) = \set{\ell(m)}{ m \in M }$.
Hereafter, we use $\M = (M, \tau, \ell)$ and $\N = (N, \theta, \alpha)$ to range over MCs and we refer to their constituents implicitly. 

\begin{definition}[Probabilistic Bisimulation~\cite{LarsenS91}] \label{def:pbisim}
An equivalence relation 
$R \subseteq M \times M$ is a \emph{probabilistic bisimulation} on $\mathcal{M}$ if
whenever $m \mathrel{R} n$, then 
\begin{enumerate}[topsep=0.5ex, noitemsep]
  \item $\ell(m) = \ell(n)$, and
  \item for all $C \in M/_R$, $\tau(m)(C) = \tau(n)(C)$.
\end{enumerate}
Two states $m,n \in M$ are \emph{probabilistic bisimilar w.r.t.\ $\M$}, written $m \sim_\M n$ if they are related by some probabilistic bisimulation on $\M$. In fact, probabilistic bisimilarity is the greatest probabilistic bisimulation. 
\end{definition}

Any bisimulation $R$ on $\M$ induces a quotient construction, the \emph{$R$-quotient of $\M$}, denoted 
$\M/_R = (M/_R, \tau/_R, \ell/_R)$, having $R$-equivalence classes as states, transition function 
$\tau/_R([m]_R)([n]_R) = \sum_{u \in [n]_R} \tau(m)(u)$, and labelling function $\ell/_R([m]_R) = \ell(m)$.
An MC $\M$ is said \emph{minimal} if it is isomorphic to its quotient w.r.t.\ 
probabilistic bisimilarity.

A $1$-bounded \emph{pseudometric} on $X$ is a function $d \colon X \times X \to [0,1]$ such that, for any $x,y,z \in X$, $d(x,x) = 0$, $d(x,y) = d(y,x)$, and $d(x,y) + d(y,z) \geq d(x,z)$. 
$1$-bounded pseudometrics on $X$ forms a complete lattice under the point-wise partial order $d \sqsubseteq d'$ iff, for all $x,y \in X$, $d(x,y) \leq d'(x,y)$.

A pseudometric is said to lift an equivalence relation if it enjoys the property that two points are at distance zero iff they are related by the equivalence.
A lifting for the probabilistic bisimilarity is provided by the \emph{bisimilarity distance} of Desharnais et al.~\cite{DesharnaisGJP04}. 
Its definition is based on the \emph{Kantorovich (pseudo)metric} on probability distributions over a finite set $X$, defined as 
\begin{equation}
\K[d]{\mu,\nu} = \min \set{\textstyle \sum_{x,y \in X} d(x,y) \cdot \omega(x,y)}{ \omega \in \coupling{\mu}{\nu}} \,,
\label{eq:Kantorivichdef}
\end{equation}
where $d$ is a (pseudo)metric on $X$ and $\coupling{\mu}{\nu}$ denotes the set of \emph{couplings} for $(\mu, \nu)$, i.e., distributions $\omega \in \D(X \times X)$ such that, 
\begin{align*}
&\textstyle \sum_{x \in X} \omega(x,y) = \nu(y) && \text{for all $y \in X$} \\
&\textstyle \sum_{y \in X} \omega(x,y) = \mu(x) && \text{for all $x \in X$.}
\end{align*}
The above condition can be equivalently stated as for all $E \subseteq X$, $\omega(E \times X) = \mu(E)$ and $\omega(X \times E) = \nu(E)$, and $\mu$ (resp.\ $\nu$) will be called left (resp.\ right) marginal of $\omega$.
\begin{remark}[Kantorovich as a Transportation Problem]
The Kantorovich metric has an intuitive interpretation as the solution of an optimization problem
usually referred to as (homogeneous) \emph{mass transportation problem}. Assume you are given 
a pile of sand and a hole we have to completely fill up with the sand. Obviously, the pile and the hole must
have the same volume. Both the pile and the hole are modeled by probability measures $\mu,\nu \in \D(X)$, with $\mu(E)$ giving a measure of how much sand is located in the pile in location $E \subseteq X$ and,
$\nu(F)$ how much sand can be piled in $F \subseteq X$. Moving the sand around should be done by
minimizing the traveling distance. In this respect a coupling can be interpreted as a transportation plan (or schedule).

A convenient way for visualizing a coupling $\omega \in \D(X \times X)$ for $(\mu,\nu)$ is by means of the so called \emph{transportation table}, with coordinates given by the supports of the measures 
$\mu$ and $\nu$, and cells containing the value $\omega(x,y)$ at coordinate $(x,y)$ (for 
convenience, when $\omega(x,y) = 0$ the cell is left blank). In this representation the condition of $\omega$ of being a coupling can be easily checked by summing up the values of the cell in each row (resp., column) and checking that it equals the value of the corresponding marginal.

To make explicit the cost of the transportation from
$x$ to $y$, the cell $(x,y)$ is further decorated with the distance $d(x,y)$, displayed in the top-left corner of the cell.

As an example, consider the MC $\M = (M, \tau, \ell)$ in Figure~\ref{fig:aggregation} and the transition probabilities 
$\tau(m_1) = \frac{1}{3} \cdot 1_{m_3} + \frac{2}{3} \cdot 1_{m_4}$ and $\tau(m_2) = \frac{1}{2} \cdot 1_{m_4} + \frac{1}{2} \cdot 1_{m_5}$. The coupling $\omega \in \coupling{\tau(m_1)}{\tau(m_2)}$ defined as
\begin{equation*}
\omega = \frac{1}{2} \cdot 1_{(m_4,m_4)} + \frac{1}{3} \cdot 1_{(m_3,m_5)} + \frac{1}{6} \cdot 1_{(m_4,m_5)}
\end{equation*}
is equivalently represented as a transportation table as follows: 
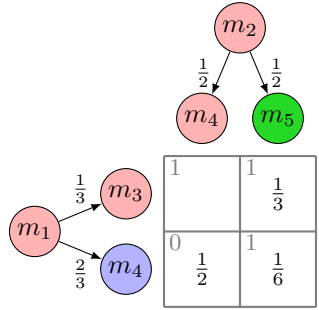
\begin{center}
\begin{tikzpicture}[baseline=(current bounding box.center), labels]
\def\h{1cm}
\def\v{1cm}
 \draw[step=1cm,gray,thick] (1,-1) grid (3,-3);
\draw 
  (0.5,-1.5) node[label,fill=red!30] (n0) {$m_3$}
  ($(n0)+(0,-1)$) node[label,fill=blue!30] (n1) {$m_4$}
  ($(n0)!0.5!(n1)+(-1.2,0)$) node[label,fill=red!30] (n2) {$m_1$}
  
  (1.5,-0.5) node[label,fill=red!30] (m1) {$m_4$}
  ($(m1)+(1,0)$) node[label,fill=green!70!gray] (m4) {$m_5$}
  ($(m1)!0.5!(m4)+(0,1.2)$) node[label,fill=red!30] (m0) {$m_2$}

  (1.15,-1.15) node[font=\small, color=gray] {$1$}
  (2.5,-1.5) node {$\frac{1}{3}$}
  (2.15,-1.15) node[font=\small, color=gray] {$1$}
  (1.5,-2.5) node {$\frac{1}{2}$}
  (1.15,-2.15) node[font=\small, color=gray] {$0$}
  (2.5,-2.5) node {$\frac{1}{6}$}
  (2.15,-2.15) node[font=\small, color=gray] {$1$}
  ;
  
  \path[-latex, font=\small]
  (m0) edge node[left] {$\frac{1}{2}$} (m1)
  	edge node[right] {$\frac{1}{2}$} (m4)
  (n2) edge node[above] {$\frac{1}{3}$} (n0)
  	edge node[below] {$\frac{2}{3}$} (n1)
  ;
\end{tikzpicture}
\end{center}
Here we consider the distance $d \colon M \times M \to [0,1]$ defined as $d(m,n) = 0$ if $m=n$, $1$ otherwise. Note that the above is an optimal coupling in the sense of \eqref{eq:Kantorivichdef}.
\end{remark}

\begin{definition}[Bisimilarity Distance]
Let $\lambda \in (0,1]$. The $\lambda$-discounted \emph{bi\-sim\-i\-larity pseudometric} on $\M$, denoted by $\dist$, is the least fixed-point of the following functional operator on $1$-bounded pseudometrics over $M$ (ordered point-wise)
\begin{equation*}
    \Psi_{\lambda}(d)(m,n) = 
    \begin{cases}
    	1 &\text{if $\ell(m) \neq \ell(n)$} \\
	\lambda \cdot \mathcal{K}(d)(\tau(m),\tau(n)) & \text{otherwise} \,.
    \end{cases}
\end{equation*}
\end{definition}
The operator $\Psi_{\lambda}$ is monotonic with respect to $\sqsubseteq$, hence, by Knaster-Tarski's fixed-point theorem, $\dist$ is well defined.

Intuitively, if two states have different labels $\dist$ considers them as ``incomparable'' (i.e., at distance $1$), otherwise their distance is given by the Kantorovich distance w.r.t.\ $\dist$ between their  
transition distributions.
The \emph{discount factor} $\lambda \in (0,1]$ controls the significance of the future steps in the measurement of the distance; if $\lambda = 1$, the distance is said \emph{undiscounted}.

The  distance $\dist$ has also a characterization based on the notion of coupling structure.
\begin{definition}[Coupling Structure] \label{def:couplstr}
A function $\C \colon M \times M \to \D(M \times M)$ is a \emph{coupling structure} for $\M$ if  
for all $m,n \in M$, $\C(m,n) \in \coupling{\tau(m)}{\tau(n)}$.
\end{definition} 
Intuitively, a coupling structure can be thought of as an MC on the cartesian product $M \times M$, 
obtained as the probabilistic combination of two copies of $\M$. 

Given a coupling structure $\C$ for $\M$ and $\lambda \in (0,1]$, let $\discr{\C}$ be the least fixed-point of the following operator on $[0,1]$-valued functions $d \colon M \times M \to [0,1]$ (ordered point-wise)
\begin{equation*}
    \Gamma^{\C}_{\lambda}(d)(m,n) = 
    \begin{cases}
    	1 &\text{if $\ell(m) \neq \ell(n)$} \\
	\lambda \sum_{u,v \in M} d(u,v) \cdot \C(m,n)(u,v) & \text{otherwise} \,.
    \end{cases}
\end{equation*}
The function $\discr{\C}$ is called \emph{$\lambda$-discounted discrepancy} of $\C$, and the value
$\discr{\C}(m,n)$ is the \mbox{$\lambda$-discounted} probability of hitting from $(m,n)$ a pair of states with different labels in $\C$. 
\begin{theorem}[Minimal coupling criterion~\cite{ChenBW12}] \label{th:mincoupling}
For arbitrary MCs $\M$ and discount factors $\lambda \in (0,1]$,
$\dist = \min \set{\discr{\C}}{ \text{$\C$ coupling structure for $\M$} }$.
\end{theorem}

As originally noted in~\cite[Lemma 10]{ChenBW12}, the (undiscounted) discrepancy can be used to bound the variational distance between trace distributions. The following lemma generalizes this result for arbitrary discount values. In the lemma we use $\Pr[\M,m]{A}$ to denote the probability that a run of the Markov chain $\M$ starting in state $m$ is in the set $A \subseteq L^\omega$. For a formal definition of $\Pr[\M,m]{A}$ and a definition of measurable subset of the set $L^\omega$ of infinite sequences over $L$, we refer the reader to, e.g.,~\cite[Chapter 10]{BaierK08}.

For $\lambda \in (0,1]$ we use $\M_\lambda$ to denote the MC $(M \uplus \{\bot \},\tau_\lambda, \ell_\lambda)$ obtained by adding to the MC $\M$ a `sink' state $\bot$ to which all states in $\M$ go with probability $(1-\lambda)$, that is, for all $m \in M$, $\tau_\lambda(m) = (1 - \lambda) 1_\bot + \lambda \tau(m)$, and $\tau_\lambda(\bot) = 1_\bot$. The `sink' state $\bot$ has label different from all other states, that is, for all $m \in M$, $\ell_\lambda(m) = \ell(m)$, and $\ell_\lambda(\bot) \neq \ell_\lambda(m)$.

\begin{lemma}\label{lem:disc-undisc}
Let $\C$ be a coupling structure for the MC $\M$. Then, for any measurable set $A \subseteq L^\omega$, discount factor $\lambda \in (0,1]$, and $m,n \in M$, $| \Pr[\M_\lambda,m]{A} - \Pr[\M_\lambda,n]{A} | \leq \discr{\C}(m,n)$. 
\end{lemma}
\begin{proof}
Let define $\C_\lambda \colon (M \uplus \{\bot \})^2 \to \D((M \uplus \{\bot \})^2)$, for arbitrary $m,n \in M$, as
\begin{align*}
&\C_\lambda(m,n) = (1-\lambda) \cdot 1_{(\bot,\bot)} + \lambda \cdot \C(m,n)
&& \C_\lambda(\bot,\bot) = 1_{(\bot,\bot)} \\
&\C_\lambda(m,\bot) = \textstyle \sum_{u \in M} \tau(m)(u) \cdot 1_{(m,\bot)} 
&& \C_\lambda(\bot,n) = \textstyle \sum_{v \in M} \tau(n)(v) \cdot 1_{(\bot,n)} \,.
\end{align*}
One can easily verify that $\C_\lambda$ is a coupling structure for $\M_\lambda$.
By~\cite[Lemma 10]{ChenBW12},
\begin{equation}
| \Pr[\M_\lambda,m]{A} - \Pr[\M_\lambda,n]{A} | \leq \discr[1]{\C_\lambda}(m,n) \,.
\label{eq:tvupper}
\end{equation}
We conclude the proof by showing that for any $m,n \in M$, $\discr[1]{\C_\lambda}(m,n) \leq  \discr{\C}(m,n)$. 

Let $d \colon (M \uplus \{\bot\}) \to [0,1]$ be defined, for $m, n \in M$, as $d(m,n) = \discr{\C}(m,n)$, $d(m,\bot) = d(\bot,n) = 1$, and $d(\bot, \bot) = 0$. We show that $d$ is a fixed point for $\Gamma^{\C_\lambda}_{1}$, that is $\Gamma^{\C_\lambda}_{1}(d)(m,n) = d(m,n)$ for all $m, n \in M \uplus \{\bot\}$. We show only the cases when $m$ and $n$ have the same label ---the other cases are immediate by definition of $\Gamma^{\C_\lambda}_{1}$ and $d$. 
\begin{trivlist}
\item Case $m = \bot$ and $n = \bot$. 
\begin{align*}
\Gamma^{\C_\lambda}_{1}(d)(\bot,\bot) 
&= \textstyle \sum_{u,v \in M \uplus \{\bot\}} d(u,v) \cdot \C_\lambda (\bot,\bot)(u,v) 
	\tag{def.\ $\Gamma^{\C_\lambda}_{1}$} \\
&= d(\bot,\bot) \,. \tag{$\C_\lambda (\bot,\bot) = 1_{(\bot,\bot)}$}
\end{align*}
\item Case $m,n \in M$ and $\ell(m) = \ell(n)$. 
\begin{align*}
\Gamma^{\C_\lambda}_{1}(d)(m,n) 
&= \textstyle \sum_{u,v \in M \uplus \{\bot\}} d(u,v) \cdot \C_\lambda (m,n)(u,v) 
	\tag{def.\ $\Gamma^{\C_\lambda}_{1}$} \\
&= \textstyle \sum_{u,v \in M} d(u,v) \cdot \C_\lambda (m,n)(u,v)  \tag{def.\ $\C_\lambda$ and $d(\bot,\bot) = 0$}\\
&= \lambda \textstyle \sum_{u,v \in M} d(u,v) \cdot \C(m,n)(u,v)  \tag{def.\ $\C_\lambda$}\\
&= \lambda \textstyle \sum_{u,v \in M}\discr{\C}(u,v) \cdot \C(m,n)(u,v)  \tag{def.\ $d$}\\
&= \Gamma^{\C}_{\lambda}(\discr{\C})(m,n) \tag{def.\ $\Gamma^{\C}_{\lambda}$} \\
&= \discr{\C}(m,n) \tag{def.\ $\discr{\C}$} \\
&= d(m,n) \,. \tag{def.\ $d$ and $m,n \in M$}
\end{align*}
\end{trivlist} 
By definition, $\discr[1]{\C_\lambda}$ is the least fixed point of $\Gamma^{\C_\lambda}_{1}$, therefore, for all $m,n \in M$, 
\begin{equation*}
\discr[1]{\C_\lambda}(m,n) \leq d(m,n) = \discr{\C}(m,n) \,.
\end{equation*}
By the above inequality and \eqref{eq:tvupper} we conclude $| \Pr[\M_\lambda,m]{A} - \Pr[\M_\lambda,n]{A} | \leq \discr{\C}(m,n)$.
\qed
\end{proof}
The following is a generalization of~\cite[Corollary 10]{ChenBW12} for arbitrary discount factors.
\begin{corollary} \label{cor:disctvupperbound}
For any measurable set $A \subseteq L^\omega$, discount factor $\lambda \in (0,1]$, and $m,n \in M$, $$| \Pr[\M_\lambda,m]{A} - \Pr[\M_\lambda,n]{A} | \leq \dist(m,n) \,.$$
\end{corollary}
\begin{proof}
Immediate consequence of Theorem~\ref{th:mincoupling} and Lemma~\ref{lem:disc-undisc}. \qed
\end{proof}

So far we have considered a single Markov chain and described a pseudometric space over its states. The above definitions can be naturally extended to reason about the distance between two MCs by considering the distance induced over their \emph{disjoint union}.

Given two MCs $\M = (M,\tau,\ell)$ and $\N = (N,\theta,\alpha)$ with $M \cap N = \emptyset$, their disjoint union, denoted as $\M \oplus \N$, is the Markov chain having state space $M \cup N$, probability transition function $\vartheta \in (M \cup N) \to \D(M \cup N)$ and labelling function $l \colon (M \cup N) \to L$ respectively defined as
\begin{align*}
&\vartheta(x) = \begin{cases}
\tau(x) & \text{if $x \in M$} \\
\theta(x) & \text{if $x \in N$}
\end{cases}
&& 
l(x) = \begin{cases}
\ell(x) & \text{if $x \in M$} \\
\alpha(x) & \text{if $x \in N$.}
\end{cases}
\end{align*}

Usually, MCs are associated with an initial state to be thought of as their initial configurations. 
In the rest of the paper when we talk about the distance between two MCs, written $\dist(\M,\N)$, we
implicitly refer to the distance between their initial states computed over the disjoint union of their MCs\footnote{Here we implicitly assume that different MCs have disjoint sets of states. This can be done without loss of generality since the behaviors that we observe are equal up-to isomorphism. Therefore, when we write e.g., $\dist(\M,\M)$, we actually compare $\M$ with another MC $\M'$ isomorphic to $\M$ but whose state space is disjoint from that of $\M$.}. Analogously, we may simply write $\Pr[\M]{A}$ instead of $\Pr[\M,m]{A}$ when $m$ is the initial state of $\M$.

\section{The Closest Bounded Approximant Problem}
\label{sec:OBA}

In this section we introduce the \emph{Closest Bounded Approximant} problem w.r.t. $\dist$ ($\CBA$), and give a characterization of it as a bilinear optimization problem. 

\begin{definition}[Closest Bounded Approximant]
Let $k \in \naturals$ and $\lambda \in (0,1]$. The \emph{closest bounded approximant problem} w.r.t. $\dist$ for an MC $\M$ is the problem of finding an MC $\N$ with at most $k$ states minimizing $\dist(\M,\N)$.
\end{definition}

Clearly, when $k$ is greater than or equal to the number of bisimilarity classes of $\M$, 
an optimal solution of $\CBA$ is the bisimilarity quotient. 
Therefore, without loss of generality, we will assume $1 \leq k < |M|$ and $\M$ to be minimal. Note that, under these assumptions $\M$ must have at least two nodes with different labels.

Let $\MC{}$ denote the set of MCs with at most $k$ states and $\MC{A}$  
its restriction to those using only labels in $A \subseteq L$.
Using this notation, the optimization problem $\CBA$ on the instance $\instance{\M,k}$ can be reformulated as 
finding an MC $\N^*$ such that
\begin{equation}
\dist(\M,\N^*) = \min \set{\dist(\M,\N)}{ \N \in \MC{} } \,.
\label{eq:AM}
\end{equation} 
In general, it is not obvious that  for arbitrary instances $\instance{\M,k}$ a minimum in \eqref{eq:AM} exists. At the end of the section, we will show that such a minimum always exists (Corollary~\ref{cor:existssolAM}).

A useful property of $\CBA$ is that an optimal solution can be found among the MCs using labels from the given MC. 
\begin{lemma}[Meaningful labels] \label{lem:usefullabels}
Let $\M$ be an MC. Then, for any $\N' \in \MC{}$ there exists $\N \in \MC{L(\M)}$ 
such that $\dist(\M,\N) \leq \dist(\M,\N')$. 
\end{lemma}
\begin{proof}
Let $\N' = (N', \theta', \alpha')$.
If $L(\N') \subseteq L(\M)$, take $\N = \N'$. Otherwise, define $\N = (N, \theta, \alpha)$ as follows:
$N = N'$, $\theta = \theta'$, and $\alpha(n) = \alpha'(n)$ if $\alpha'(n) \in L(\M)$, otherwise $\alpha(n) = \ell(m_0)$, where $m_0$ is the initial state of $\M$. The initial state of $\N$ is the one of $\N'$.
Clearly, $\N \in \MC{}$ and $L(\N) \subseteq L(\M)$.

Let $\mathcal{A} = \M \oplus \N$ and $\mathcal{B} = \M \oplus \N'$, 
we prove $\dist(\M,\N) \leq \dist(\M,\N')$ by showing a stronger statement: $\dist^\mathcal{A} \sqsubseteq \dist^{\mathcal{B}}$. By Knaster-Tarski's fixed-point theorem, it suffices to show 
$\Psi_{\lambda}^{\mathcal{A}}(\dist^{\mathcal{B}}) \sqsubseteq \dist^{\mathcal{B}}$.
Let $u,v \in M \uplus N$. When $u$ and $v$ have different labels in $\mathcal{B}$, 
then, $\Psi_{\lambda}^{\mathcal{A}}(\dist^{\mathcal{B}})(u,v) \leq 1 = \dist^{\mathcal{B}}(u,v)$ follows by definition of $\Psi_{\lambda}$ and the fact that $\dist^{\mathcal{B}} = \Psi_{\lambda}^{\mathcal{B}}(\dist^{\mathcal{B}})$.
Assume that $u$ and $v$ have the same label in $\mathcal{B}$. Then, by construction of $\N$ (i.e, by definition of $\alpha'$), $u$ and $v$  have the same label in $\mathcal{A}$.
Since $\N$ and $\N'$ have the same transition distribution function, we obtain $\Psi_{\lambda}^{\mathcal{A}}(\dist^{\mathcal{B}})(u,v) = \dist^{\mathcal{B}}(u,v)$. \qed
\end{proof}

In the following, fix $\instance{\M,k}$ as instance of $\CBA$, let $m_0 \in M$ be the initial state of $\M$.
By Lemma~\ref{lem:usefullabels}, Theorem~\ref{th:mincoupling} and Knaster-Tarski fixed-point theorem
\begin{align}
  \inf &\set{\dist(\M,\N)}{ \N \in \MC{} } = {} \\
  &= \inf \set{\discr{\C}(\M,\N)}{ \text{$\N \in \MC{L(\M)}$ and $\C \in \coupling{\M}{\N}$} } 
  	\label{eq:CBAdiscr} \\
  &= \inf \set{d(\M,\N)}{ \text{$\N \in \MC{L(\M)}$, $\C \in \coupling{\M}{\N}$, 
  	and $\Gamma^\C_\lambda(d) \sqsubseteq d$} } \,,
\end{align}
where $\coupling{\M}{\N}$ denotes the set of all coupling structures for the disjoint union of $\M$ and $\N$.
This simple change in perspective yields a translation of the problem of computing 
the optimal value of $\CBA$ to the bilinear program in Figure~\ref{fig:BMI}. 

\begin{figure}[t]
\hrule \vspace{-1ex}
\begin{align}
\text{mimimize}\;\;& d_{m_0,n_0}&  \notag \\
\text{such that}\;\;
&\lambda \textstyle\sum_{(u,v) \in M \times N} c^{m,n}_{u,v} \cdot d_{u,v} \leq d_{m,n} 
	&& \text{$m \in M$, $n \in N$} \label{eq:prefix} \\
&1 - \alpha_{n,l} \leq d_{m,n} \leq 1 && \text{$n \in N$, $l \in L(\M)$, $\ell(m) \neq l$} 
	\label{eq:bound} \\
&\alpha_{n,l} \cdot \alpha_{n,l'} = 0 && \text{$n \in N$, $l,l' \in L(\M)$, $l \neq l'$} \label{eq:labelneq1}\\
& \textstyle \sum_{l \in L(\M)} \alpha_{n,l} = 1 && \text{$n \in N$}  
\label{eq:labelneq3} \\
&\textstyle\sum_{v \in N} c^{m,n}_{u,v} = \tau(m)(u)  && \text{$m,u \in M$, $n \in N$} 
\label{eq:leftmarginal} \\ 
&\textstyle\sum_{u \in M} c^{m,n}_{u,v} = \theta_{n,v} && \text{$m \in M$, $n,v \in N$}  
\label{eq:rightmarginal} \\
&c^{m,n}_{u,v} \geq 0 && \text{$m,u \in M$, $n,v \in N$} \label{eq:nonnegative}
\end{align}
\hrule
\caption{Characterization of $\CBA$ as a bilinear optimization problem.} \label{fig:BMI}
\end{figure}


In our encoding, $N = \{n_0, \dots, n_{k-1}\}$ are the states of an arbitrary Markov chain with $k$ states $\N = (N,\theta,\alpha) \in \MC{}$.
and $n_0$ is the initial one. The variable $\theta_{n,v}$ is used to encode the transition probability $\theta(n)(v)$. Hence, a feasible solution satisfying (\ref{eq:leftmarginal}--\ref{eq:nonnegative}) will have the variable $c^{m,n}_{u,v}$ representing the value $\C(m,n)(u,v)$ for a coupling structure $\C \in \coupling{\M}{\N}$. 
An assignment for the variables $\alpha_{n,l}$ satisfying (\ref{eq:labelneq1}--\ref{eq:labelneq3}) encodes (uniquely) a labeling function $\alpha \colon N \to L(\M)$ satisfying the following property:
\begin{align}
\text{for all } n \in N, l \in L(\M)
&&
\alpha_{n,l} = 1 \quad \text{ iff } \quad \alpha(n) = l  \,.
\label{eq:alpha}
\end{align}
The constraint \eqref{eq:labelneq1} models the fact that each node $n \in N$ is assigned at most to one label $l \in L(\M)$, and the constraint \eqref{eq:labelneq3} ensures that each node is assigned to at least one label in $L(\M)$.
Conversely, any labeling $\alpha \colon N \to L(\M)$ admits an assignment of the variables $\alpha_{n,l}$ satisfying (\ref{eq:labelneq1}--\ref{eq:labelneq3}) and \eqref{eq:alpha}.
Finally, an assignment for the variables $d_{m,n}$ satisfying the constraints (\ref{eq:prefix}--\ref{eq:bound}) represents a prefix point of $\Gamma^\C_\lambda$. Note that \eqref{eq:bound} guarantees that $d_{m,n} = 1$ whenever $\alpha(n) \neq \ell(m)$ ---indeed, by \eqref{eq:labelneq1}, $\alpha_{n,l} = 0$ iff $\alpha(n) \neq \ell(m)$.

Let $F_\lambda\instance{\M,k}$ denote the bilinear optimization problem in Figure~\ref{fig:BMI}. Directly from the
arguments stated above we obtain the following result.
\begin{theorem} \label{thm:optvalue}
$\inf \set{\dist(\M,\N)}{ \N \,{\in}\, \MC{} }$ is the optimal value of $F_\lambda\instance{\M,k}$.
\end{theorem}

\begin{corollary} \label{cor:existssolAM}
Any instance of $\CBA$ admits an optimal solution. 
\end{corollary}
\begin{proof}
We show that $\dist(\M,\N^*) = \inf \set{\dist(\M,\N)}{ \N \in \MC{} }$ for some $\N^* \in \MC{}$.
Let $h$ be the number of variables in $F_\lambda\instance{\M,k}$. The constraints (\ref{eq:bound}--\ref{eq:nonnegative}) describe a compact subset of $\mathbb{R}^h$ ---it is an intersection of closed sets bounded by $[0,1]^h$. The objective function of $F_\lambda\instance{\M,k}$ is linear, hence the infimum is attained by a feasible solution. The thesis follows by Theorem~\ref{thm:optvalue}. \qed
\end{proof}

The next example shows that even by starting with an MC with rational transition probabilities, optimal solutions
for $\CBA$ may have irrational transition probabilities. 
\begin{example} \label{ex:minimalirrational}
Consider the MC $\M$ depicted below, with initial state $m_0$ and labels represented by distinct colors.
We claim that the MC $\N_{xy}$ depicted below, with initial state $n_0$ and parameters 
$x = \frac{1}{30} \left(10+\sqrt{163}\right)$, $y = \frac{21}{100}$, is an optimal solution of $\CBA[1]$ on input $\instance{\M,3}$.
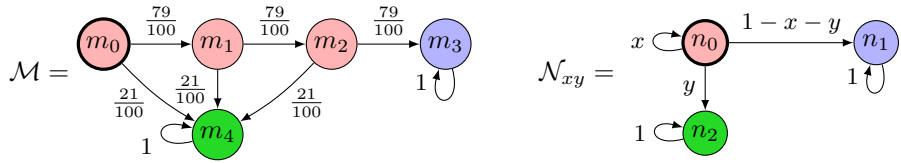
\def\h{1.5cm}
\def\v{1.2cm}
\begin{align*}
\M = 
\tikz[labels, baseline={(current bounding box.center)}, 
n1/.style={fill=red!30}, n2/.style={fill=blue!30}, n3/.style={fill=green!70!gray}]{ 
  \draw (0,0) node[label, initial, n1] (m0) {$m_0$}; 
  \draw ($(m0)+(right:\h)$) node[label, n1] (m1) {$m_1$};
  \draw ($(m1)+(right:\h)$) node[label, n1] (m2) {$m_2$};
  \draw ($(m2)+(right:\h)$) node[label, n2] (m3) {$m_3$};
  \draw ($(m0)!0.5!(m2)+(down:\v)$) node[label, n3] (m4) {$m_4$};
  \path[-latex, font=\small]
    (m0) edge node[above] {$\frac{79}{100}$} (m1)
    	    edge[bend right=8] node[left, yshift=-1mm] {$\frac{21}{100}$} (m4)
    (m1) edge node[above] {$\frac{79}{100}$} (m2)
            edge node[left] {$\frac{21}{100}$} (m4)
    (m2) edge node[above] {$\frac{79}{100}$} (m3)
            edge[bend left=8] node[right, yshift=-1mm] {$\frac{21}{100}$} (m4)
    (m3) edge[loop below] node[very near end, left] {$1$} (m3)
    (m4) edge[loop left] node[left, yshift=-1.5mm] {$1$} (m4)
    ;
}
&&
\N_{xy} = 
\tikz[labels, baseline={(current bounding box.center)}, 
n1/.style={fill=red!30}, n2/.style={fill=blue!30}, n3/.style={fill=green!70!gray}]{ 
  \draw (0,0) node[label, initial, n1] (n0) {$n_0$}; 
  \draw ($(n0)+(right:1.5*\h)$) node[label, n2] (n1) {$n_1$};
  \draw ($(n0)+(down:\v)$) node[label, n3] (n2) {$n_2$};
  \path[-latex, font=\small]
    (n0) edge node[left] {$y$} (n2)
           edge node[above] {$1-x-y$} (n1)
           edge[loop left] node[left] {$x$} (n0)
    (n1) edge[loop below] node[very near end, left] {$1$} (n1)
    (n2) edge[loop left] node[left] {$1$} (n2)
    ;
}
\end{align*}
Since the distance $\dist[1](\M,\N_{xy}) = \frac{436}{675}-\frac{163 \sqrt{163}}{13500} \approx 0.49$ is
irrational, by~\cite[Proposition 13]{ChenBW12}, any optimal solution must have some irrational transition probability.

Now we prove that the above is indeed an optimal solution for $\CBA[1]$ on input $\instance{\M,3}$. Assume by contradiction that there exists $\N^* \in \MC[3]{}$ s.t.\ $\dist[1](\M,\N^*) <  \dist[1](\M,\N_{xy})$. Without loss of generality, by Lemma~\ref{lem:usefullabels} we can assume that $L(\N^*) \subseteq L(\M)$. 

If $L(\N^*) = L(\M)$, then $\N^*$ must be an MC of the form $\N_{zw}$ for some of $z,w \in [0,1]$ such that $z + w \leq 1$. Thus,
\begin{align}
\dist[1](\M,\N^*) &\geq \min \set{ \dist[1](\M,\N_{zw})}{z,w \geq 0,\, z + w \leq 1} \label{eq:example} \\
&= \min \set{ \discr[1]{\C}(m_0,n_0) }{\C \in \coupling{\M}{\N_{zw}},\, z,w \geq 0,\, z + w \leq 1} \,.
\tag{Thm.~\ref{th:mincoupling}}
\end{align}
Consider an arbitrary coupling structure $\C \in \coupling{\M}{\N_{zw}}$ for some $z,w \geq 0$ such that 
$z + w \leq 1$. By definition of $\discr[1]{\C}$, for any $i \in \{0,1,2\}$,
\begin{equation}
\discr[1]{\C}(m_i,n_0) = A_i + \textstyle\sum_{j = 0}^1\discr[1]{\C}(m_{i+1},n_j) \cdot \C(m_i, n_0)(m_{i+1},n_j)  \,,
\label{eq:discrexample}
\end{equation}
where $A_i = \C(m_i, n_0)(m_{i+1},n_2) + \C(m_i, n_0)(m_4,n_0) + \C(m_i, n_0)(m_4,n_1)$. 
The constraints on the marginals require that, for all $i \in \{0, 1, 2\}$,
\begin{align*}
w &= \C(m_i, n_0)(m_{i+1},n_2) + \C(m_i, n_0)(m_4,n_2) \,, \text{ and} \\
\frac{21}{100} &= \C(m_i, n_0)(m_4,n_0) + \C(m_i, n_0)(m_4,n_1) + \C(m_i, n_0)(m_4,n_2) \,.
\end{align*}

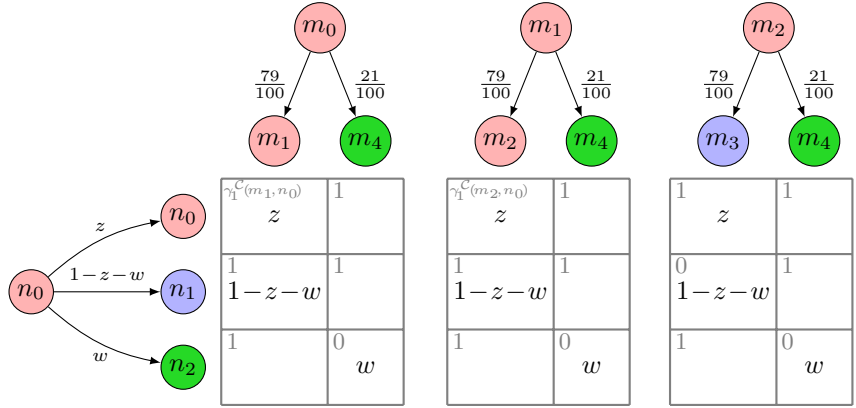
\begin{figure}
\centering
\begin{tikzpicture}[baseline=(current bounding box.center), labels]
\def\h{1cm}
\def\v{1cm}

\draw 
  (0.5,-1.5) node[label,fill=red!30] (n0) {$n_0$}
  ($(n0)+(0,-1)$) node[label,fill=blue!30] (n1) {$n_1$}
  ($(n1)+(0,-1)$) node[label,fill=green!70!gray] (n2) {$n_2$}
  ($(n0)!0.5!(n2)+(left:2)$) node[label,fill=red!30] (n01) {$n_0$}
  ;
  
\coordinate (g1) at ($(n0)!0.5!(1.5,-0.5)$);
\coordinate (g2) at ($(g1)+(right:1.4)$);
\coordinate (g3) at ($(g2)+(right:1)$);
  
\draw
  ($($(g1)!0.5!(g2)$)+(up:0.5)$) node[label,fill=red!30] (m1) {$m_1$}
  ($($(g2)!0.5!(g3)$)+(up:0.5)$) node[label,fill=green!70!gray] (m4) {$m_4$}
  ($(m1)!0.5!(m4)+(0,1.5)$) node[label,fill=red!30] (m0) {$m_0$}
  ;

\draw  
  ($(g1)!0.5!($(g2)+(down:1)$)$) node {$z$}
  ($(g1)+(0.55,-0.15)$) node[font=\tiny, color=gray] {$\discr[\!1]{\C}\!(\!m_1\!,\!n_0\!)$}
  ($(g2)+(0.15,-0.15)$) node[font=\small, color=gray] {$1$}
  ($($(g1)+(down:1)$)!0.5!($(g2)+(down:2)$)$) node {$1\!-\!z\!-\!w$}
  ($(g1)+(0.15,-1.15)$) node[font=\small, color=gray] {$1$}
  ($(g2)+(0.15,-1.15)$) node[font=\small, color=gray] {$1$}
  ($(g1)+(0.15,-2.15)$) node[font=\small, color=gray] {$1$}
  ($($(g2)+(down:2)$)!0.5!($(g3)+(down:3)$)$) node {$w$}
  ($(g2)+(0.15,-2.15)$) node[font=\small, color=gray] {$0$}
  ;
  
  \draw[gray,thick] (g1) -- ($(g1)+(down:3)$)
  	(g2) -- ($(g2)+(down:3)$)
	(g3) -- ($(g3)+(down:3)$)
	(g1) -- (g3)
	($(g1)+(down:1)$) -- ($(g3)+(down:1)$)
	($(g1)+(down:2)$) -- ($(g3)+(down:2)$)
	($(g1)+(down:3)$) -- ($(g3)+(down:3)$)
	;
  
  \path[-latex, font=\small]
  (m0) edge node[left] {$\frac{79}{100}$} (m1)
  	edge node[right] {$\frac{21}{100}$} (m4)
  ;
  \path[-latex, font=\scriptsize]
  (n01) edge[bend left=15] node[above] {$z$} (n0)
  	edge node[above] {$1\!-\!z\!-\!w$} (n1)
	edge[bend right=15] node[below] {$w$} (n2)
  ;
\end{tikzpicture}
\quad
\begin{tikzpicture}[baseline=(current bounding box.center), labels]
  
\coordinate (g1) at (1,-1);
\coordinate (g2) at ($(g1)+(right:1.4)$);
\coordinate (g3) at ($(g2)+(right:1)$);
  
\draw
  ($($(g1)!0.5!(g2)$)+(up:0.5)$) node[label,fill=red!30] (m1) {$m_2$}
  ($($(g2)!0.5!(g3)$)+(up:0.5)$) node[label,fill=green!70!gray] (m4) {$m_4$}
  ($(m1)!0.5!(m4)+(0,1.5)$) node[label,fill=red!30] (m0) {$m_1$}
  ;

\draw  
  ($(g1)!0.5!($(g2)+(down:1)$)$) node {$z$}
  ($(g1)+(0.55,-0.15)$) node[font=\tiny, color=gray] {$\discr[\!1]{\C}\!(\!m_2\!,\!n_0\!)$}
  ($(g2)+(0.15,-0.15)$) node[font=\small, color=gray] {$1$}
  ($($(g1)+(down:1)$)!0.5!($(g2)+(down:2)$)$) node {$1\!-\!z\!-\!w$}
  ($(g1)+(0.15,-1.15)$) node[font=\small, color=gray] {$1$}
  ($(g2)+(0.15,-1.15)$) node[font=\small, color=gray] {$1$}
  ($(g1)+(0.15,-2.15)$) node[font=\small, color=gray] {$1$}
  ($($(g2)+(down:2)$)!0.5!($(g3)+(down:3)$)$) node {$w$}
  ($(g2)+(0.15,-2.15)$) node[font=\small, color=gray] {$0$}
  ;
  
  \draw[gray,thick] (g1) -- ($(g1)+(down:3)$)
  	(g2) -- ($(g2)+(down:3)$)
	(g3) -- ($(g3)+(down:3)$)
	(g1) -- (g3)
	($(g1)+(down:1)$) -- ($(g3)+(down:1)$)
	($(g1)+(down:2)$) -- ($(g3)+(down:2)$)
	($(g1)+(down:3)$) -- ($(g3)+(down:3)$)
	;
  
  \path[-latex, font=\small]
  (m0) edge node[left] {$\frac{79}{100}$} (m1)
  	edge node[right] {$\frac{21}{100}$} (m4)
  ;
\end{tikzpicture}
\quad
\begin{tikzpicture}[baseline=(current bounding box.center), labels]
  
\coordinate (g1) at (1,-1);
\coordinate (g2) at ($(g1)+(right:1.4)$);
\coordinate (g3) at ($(g2)+(right:1)$);
  
\draw
  ($($(g1)!0.5!(g2)$)+(up:0.5)$) node[label,fill=blue!30] (m1) {$m_3$}
  ($($(g2)!0.5!(g3)$)+(up:0.5)$) node[label,fill=green!70!gray] (m4) {$m_4$}
  ($(m1)!0.5!(m4)+(0,1.5)$) node[label,fill=red!30] (m0) {$m_2$}
  ;

\draw  
  ($(g1)!0.5!($(g2)+(down:1)$)$) node {$z$}
  ($(g1)+(0.15,-0.15)$) node[font=\small, color=gray] {$1$}
  ($(g2)+(0.15,-0.15)$) node[font=\small, color=gray] {$1$}
  ($($(g1)+(down:1)$)!0.5!($(g2)+(down:2)$)$) node {$1\!-\!z\!-\!w$}
  ($(g1)+(0.15,-1.15)$) node[font=\small, color=gray] {$0$}
  ($(g2)+(0.15,-1.15)$) node[font=\small, color=gray] {$1$}
  ($(g1)+(0.15,-2.15)$) node[font=\small, color=gray] {$1$}
  ($($(g2)+(down:2)$)!0.5!($(g3)+(down:3)$)$) node {$w$}
  ($(g2)+(0.15,-2.15)$) node[font=\small, color=gray] {$0$}
  ;
  
  \draw[gray,thick] (g1) -- ($(g1)+(down:3)$)
  	(g2) -- ($(g2)+(down:3)$)
	(g3) -- ($(g3)+(down:3)$)
	(g1) -- (g3)
	($(g1)+(down:1)$) -- ($(g3)+(down:1)$)
	($(g1)+(down:2)$) -- ($(g3)+(down:2)$)
	($(g1)+(down:3)$) -- ($(g3)+(down:3)$)
	;
  
  \path[-latex, font=\small]
  (m0) edge node[left] {$\frac{79}{100}$} (m1)
  	edge node[right] {$\frac{21}{100}$} (m4)
  ;
\end{tikzpicture}

\caption{Family of coupling structures for $\M$ and $\N_{z,w}$ with $z \in [0,\frac{79}{100}]$ and $w = \frac{21}{100}$(\cf~Example~\ref{ex:minimalirrational}).}
\label{fig:example1}
\end{figure}

Consequently, for minimizing $\discr[1]{\C}(m_i,n_0)$, we shall choose the values of $\C(m_i, n_0)$ in a way that makes $A_i = 0$ (hence $\C(m_i, n_0)(m_4,n_2) = \frac{21}{100}$). 
Therefore, we can restrict our attention to the family of coupling structures 
depicted in Figure~\ref{ex:minimalirrational} 
\begin{equation}
\dist[1](\M,\N^*) \geq \min \set{ \discr[1]{\C}(m_0,n_0) }{\C \in \coupling{\M}{\N_{zw}},\, 0 \leq z \leq \frac{79}{100}, \,w = \frac{21}{100} } \,. \label{eq:fixw}
\end{equation} 
We compute $\discr[1]{\C}(m_0,n_0)$ backwards
\begin{align*}
\discr[1]{\C}(m_4, n_2) &= \discr[1]{\C}(m_3, n_1) = 0 \\
\discr[1]{\C}(m_2, n_0) &= z \\
\discr[1]{\C}(m_1, n_0) &= z \cdot \discr[1]{\C}(m_2, n_0) + 1 - z - w \\
\discr[1]{\C}(m_0, n_0) &= z \cdot \discr[1]{\C}(m_1, n_0) + 1 - z - w
\end{align*}
thus, the inequality \eqref{eq:fixw} simplifies to
\begin{equation*}
\dist[1](\M,\N^*) \geq \min \left\{z^3 - z^2 - \frac{21}{100} z + \frac{79}{100} ~ \left| ~ 0 \leq z \leq \frac{79}{100} \right. \right\} 
\end{equation*}
The minimum value of the above is achieved at $z = \frac{1}{30} \left(10+\sqrt{163}\right)$. This contradicts the initial assumption $\dist[1](\M,\N^*) < \dist[1](\M,\N_{xy})$ excluding the hypothesis $L(\N^*) = L(\M)$.

Let's consider the case $L(\N^*) \subsetneq L(\M)$ and analyse three possible sub-cases.

If $\ell(m_0) \notin L(\N^*)$, then no state in $\N$ has the same label as the initial state of $\M$. Therefore, by definition of $\dist[1]$ we have $\dist[1](\M,\N^*) = 1$, which is greater than $\dist[1](\M,\N_{xy})$.
 
If $\ell(m_3) \notin L(\N^*)$ we have 
\begin{align*}
\dist[1](\M,\N^*) 
&\geq |\Pr[\M]{\ell(m_3)L^\omega} - \Pr[\N^*]{\ell(m_3)L^\omega}| 
	\tag{Corollary~\ref{cor:disctvupperbound}} \\ 
&= \Pr[\M]{L^*\ell(m_3)L^\omega} \tag{$\Pr[\N^*]{\ell(m_3)L^\omega} = 0$} \\
&= \left(\frac{79}{100}\right)^3  \tag{def.\ \Pr[\M]{}} \\
&> \dist[1](\M,\N_{xy}) \,.
\end{align*}

Analogously, if $\ell(m_4) \notin L(\N^*)$, 
\begin{equation*}
\dist[1](\M,\N^*) \geq \Pr[\M]{L^*\ell(m_4)L^\omega} = \frac{21}{100}  \sum_{i = 0}^{2} \left(\frac{79}{100}\right)^i > \dist[1](\M,\N_{xy}) \,.
\end{equation*}
Therefore $\N_{xy}$ is an optimal solution. 
\qed
\end{example}

\section{The Bounded Approximant Threshold Problem}
\label{sec:comeplexityAM}

The \emph{Bounded Approximant problem} w.r.t.\ $\dist$ ($\BA$) is the threshold decision problem
of $\CBA$, that, given MC $\M$, integer $k \geq 1$, and rational $\epsilon \geq 0$, asks whether 
there exists $\N \in \MC{}$ such that \mbox{$\dist(\M,\N) \leq \epsilon$}.

\smallskip
From the characterization of $\CBA$ as a bilinear optimization problem (Theorem~\ref{thm:optvalue}) we immediately
get the following complexity upper-bound for $\BA$.
\begin{theorem} \label{th:AMT-PSPACE}
For any $\lambda \in (0,1]$, $\BA$ is in PSPACE.
\end{theorem}
\begin{proof}
By Theorem~\ref{thm:optvalue}, deciding an instance $\instance{\M,k,\epsilon}$ of $\BA$ can be encoded as a decision problem for the existential theory of the reals, namely, checking the feasibility of the constraints (\ref{eq:bound}--\ref{eq:nonnegative}) in conjunction with $d_{m_0,n_0} \leq \epsilon$. The encoding is polynomial in the size of $\instance{\M,k,\epsilon}$, thus it can be solved in PSPACE (\cf\ Canny~\cite{Canny88}). \qed
\end{proof}

In the rest of the section we provide a complexity lower-bound for $\BA$, by showing that $\BA$ is 
NP-hard via a reduction from \VC.
Recall that, a vertex cover of an undirected graph $G$ is a subset $C$ of vertices such that every edge in $G$ has at least one endpoint in $C$. Given a graph $G$ and a positive integer $h$, 
the \VC\ problem asks if $G$ has a cover of size at most $h$.

The following lemma provides a lower-bound on the $\lambda$-discounted bisimilarity distance between $\M$ and any $\N \in \MC{}$.

\begin{lemma}\label{lem:lowerbound}
For all $m \in M$ and $n \in N$, $\dist(m,n) \geq \lambda \cdot \tau(m)(\set{u \in M}{\ell(u) \notin L(\N)})$.
\end{lemma}
\begin{proof}
The thesis holds trivially when $\ell(m) \neq \alpha(n)$, since $\dist(m,n) = 1$. 

Let $\ell(m) = \alpha(n)$, and $M' = \set{u \in M}{\ell(u) \notin L(\N)}$, then the following hold
\begin{align*}
\dist(m,n) &= \textstyle \lambda \sum_{u \in M} \sum_{v \in N} \dist(u,v) \cdot \omega(u,v) \tag{for some $\omega \in \coupling{\tau(m)}{\theta(n)}$} \\
&\geq \textstyle \lambda \sum_{u \in M'} \sum_{v \in N} \dist(u,v) \cdot \omega(u,v) 
\tag{by $M' \subseteq M$} \\
&= \textstyle \lambda \sum_{u \in M'} \sum_{v \in N} \omega(u,v) 
\tag{$\dist(u,v) = 1$ for all $u \in M'$ and $n \in N$} \\
&= \lambda \cdot \tau(m)(M') \tag{by $\omega \in \coupling{\tau(m)}{\theta(n)}$}
\end{align*}
\qed
\end{proof}

We are now ready to present the main result of this session. 
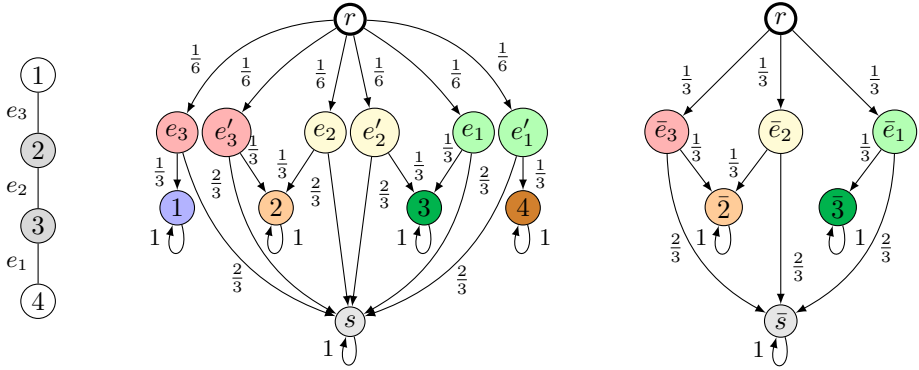
\begin{figure}[t]
\centering
\hfill
\begin{tikzpicture}[baseline=(current bounding box.center), labels]
\def\v{1cm}
\draw 
  (0,0) node[label] (n1) {$1$}
  ($(n1)+(down:\v)$) node[cover] (n2) {$2$}
  ($(n2)+(down:\v)$) node[cover] (n3) {$3$}
  ($(n3)+(down:\v)$) node[label] (n4) {$4$}
  ;
\path[font=\small]
  (n1) edge node[left] {$e_3$} (n2)
  (n2) edge node[left] {$e_2$} (n3)
  (n3) edge node[left] {$e_1$} (n4)
;
\end{tikzpicture}
\qquad\;\;
\begin{tikzpicture}[baseline=(current bounding box.center), labels]
\def\h{0.65cm}
\def\v{1cm}
\draw 
  
  (0,0) node[label,fill=red!30] (a) {$e_3$}
  ($(a)+(\h,0)$) node[label,fill=red!30] (a1) {$e'_3$}
  
  ($(a)+(6*\h,0)$) node[label,fill=green!30] (c) {$e_1$}
  ($(c)+(\h,0)$) node[label,fill=green!30] (c1) {$e'_1$}
  
  ($(a)!0.5!(c)$) node[label,fill=yellow!20] (b) {$e_2$}
  ($(b)+(\h,0)$) node[label,fill=yellow!20] (b1) {$e'_2$}
  
  ($(a)!0.5!(c1)+(0,1.5*\v)$) node[initial] (r) {$r$}
  
  ($(a)+(0,-\v)$) node[label,fill=blue!30] (n1) {$1$}
  ($(a1)!0.5!(b)+(0,-\v)$) node[label,fill=orange!40] (n2) {$2$}
  ($(b1)!0.5!(c)+(0,-\v)$) node[label,fill=green!70!blue] (n3) {$3$}
  ($(c1)+(0,-\v)$) node[label,fill=orange!65!gray] (n4) {$4$}
   
  ($(n1)!0.5!(n4)+(0,-1.5\v)$) node[label,fill=gray!20] (s) {$s$}
  ($(r)+(3*\h,-3*\v)+(right:0.5)$) coordinate (s1)
  ;

\path[-latex, font=\small]
  (r) edge[bend right] node[near end, yshift=2mm,left] {$\frac{1}{6}$} (a)
  (r) edge[bend right=10] node[near end, yshift=2mm,left] {$\frac{1}{6}$} (a1)
  (r) edge node[left] {$\frac{1}{6}$} (b)
  (r) edge node[right] {$\frac{1}{6}$} (b1)
  (r) edge[bend left=10] node[near end, yshift=2mm, right] {$\frac{1}{6}$} (c)
  (r) edge[bend left] node[near end, yshift=2mm, right] {$\frac{1}{6}$} (c1)
  
  (a) edge node[left] {$\frac{1}{3}$} (n1)
  (a1) edge node[above] {$\frac{1}{3}$} (n2)
  (b) edge node[yshift=1mm,left] {$\frac{1}{3}$} (n2)
  (b1) edge node[yshift=1mm,right] {$\frac{1}{3}$} (n3)
  (c) edge node[above] {$\frac{1}{3}$} (n3)
  (c1) edge node[right] {$\frac{1}{3}$} (n4)

  (n1) edge[loop below] node[very near end, left] {$1$} (n1)
  (n2) edge[loop below] node[very near start, right] {$1$} (n2) 
  (n3) edge[loop below] node[very near end, left] {$1$} (n3)
  (n4) edge[loop below] node[very near start, right] {$1$} (n4)
  
  (s) edge[loop below] node[very near end, left] {$1$} (s)
  
  (a) edge[bend right=27] node[midway, below] {$\frac{2}{3}$} (s)
  (a1) edge[bend right=27] node[very near start, left] {$\frac{2}{3}$} (s)
  (b) edge node[near start, left] {$\frac{2}{3}$} (s)
  (b1) edge node[near start, right] {$\frac{2}{3}$} (s)
  (c) edge[bend left=27] node[very near start, right] {$\frac{2}{3}$} (s) 
  (c1) edge[bend left=27] node[midway, below] {$\frac{2}{3}$} (s)   
  ;
\end{tikzpicture}
\qquad\;\;
\begin{tikzpicture}[baseline=(current bounding box.center), labels]
\def\h{0.6cm}
\def\v{1cm}
\draw 
  
  (0,0) node[label,label,fill=red!30] (a) {$\bar e_3$}
  ($(a)+(5*\h,0)$) node[label,label,fill=green!30] (c) {$\bar e_1$}
  ($(a)!0.5!(c)$) node[label,fill=yellow!20] (b) {$\bar e_2$}
  
  ($(a)!0.5!(b)+(0,-\v)$) node[label,fill=orange!40] (n2) {$\bar 2$}
  ($(b)!0.5!(c)+(0,-\v)$) node[label,fill=green!70!blue] (n3) {$\bar 3$}
  
  ($(n2)!0.5!(n3)+(0,-1.5\v)$) node[label,fill=gray!20] (s) {$\bar s$}
  ($(r)+(3*\h,-3*\v)+(right:0.5)$) coordinate (s1)
  
  ($(a)!0.5!(c)+(0,1.5*\v)$) node[initial] (r) {$r$}
  ;

\path[-latex, font=\small]
  (r) edge node[near end, yshift=2mm, left] {$\frac{1}{3}$} (a)
  (r) edge node[left] {$\frac{1}{3}$} (b)
  (r) edge node[near end, yshift=2mm, right] {$\frac{1}{3}$} (c)
  
  (a) edge node[above] {$\frac{1}{3}$} (n2)
  (b) edge node[yshift=1mm, left] {$\frac{1}{3}$} (n2)
  (c) edge node[above] {$\frac{1}{3}$} (n3)

  (n2) edge[loop below] node[very near end, left] {$1$} (n2)
  (n3) edge[loop below] node[very near start, right] {$1$} (n3)
  
  (s) edge[loop below] node[very near end, left] {$1$} (s)
  
  (a) edge[bend right] node[midway, left] {$\frac{2}{3}$} (s)
  (b) edge node[near end, right] {$\frac{2}{3}$} (s)
  (c) edge[bend left] node[midway, right] {$\frac{2}{3}$} (s)   
  ;
\end{tikzpicture}
\caption{(Left) An undirected graph $G$; (Center) The MC $\M_G$ associated to the graph $G$; 
(Right) The MC $\M_C$ associated to the vertex cover $C = \{2,3\}$ of $G$. (see Thm.~\ref{th:NPhardAM}).}
\label{fig:ReductionFromVertexCover}
\end{figure}

\begin{theorem} \label{th:NPhardAM}
For any $\lambda \in (0,1]$, $\BA$ is NP-hard.
\end{theorem}
\begin{proof}
We provide a polynomial-time many-one reduction from \VC. 

Let $\instance{G=(V,E), h}$ be an instance of \VC\ and let $m = |E|$. 
Without loss of generality we assume $m \geq 2$. 

From $G$ we construct the MC $\M_G = (M, \tau, \ell)$ having the following states:
a root state $r$ (thought of as the initial state); a sink state $s$; a state $v$ for each vertex in $V$; and two `twin' states $e$ and $e'$ for each edge $e \in E$. 
In $\M_G$ each pair of twin edge states have the same label, all other nodes have pairwise distinct labels. 
The sink state $s$ and all vertex states loop to themselves with probability $1$; the root state $r$ goes with uniform probability $\frac{1}{2m}$ to each edge state; for each edge $e = (v_1,v_2)$ in $G$, the state $e$ (resp.\ $e'$) in $\M_G$ goes with probability $\frac{1}{m}$ to $v_1$ (resp.\ $v_2$) and probability $1 - \frac{1}{m}$ to the sink state $s$ (\cf\ Figure~\ref{fig:ReductionFromVertexCover} for an example of the construction of $\M_G$).
 Next we show that
\begin{align*}
  \instance{G,h} \in \VC
  &&
  \text{iff}
  &&
  \textstyle
  \instance{\M_G, m+h+2, \frac{\lambda^2}{2 m}} \in \BA \,.
\end{align*}

($\Rightarrow$) Let $C$ be a $h$-vertex cover of $G$.
Construct $\M_C \in \MC[m+h+2]{}$ by taking a copy of $\M_G$, removing all vertex states in $V \setminus C$, then removing one twin edge state for each edge in $E$ making sure to keep those which are going with probability greater than zero to some vertex state in $C$ ---if both endpoints of the edge are in the cover $C$, just pick one twin edge state at random. Finally, redistribute uniformly the transition probabilities of the root state over the remaining edge states (\cf\ Figure~\ref{fig:ReductionFromVertexCover}).

Next we show that $\dist(\M_G, \M_C) \leq \frac{\lambda^2}{2m}$.
For convenience, the states in $\M_C$ will be marked with a bar to distinguish them from their counterpart in $\M_G$.
By construction of $\M_G,\M_C$, for each $e \in E$, $\dist(e,\bar{e}) + \dist(e',\bar{e}) = \frac{\lambda}{m}$ since either $e$ or $e'$ is at distance $0$ from $\bar{e}$ while the other state differs from $\bar{e}$ only for the transition to a vertex state.
Therefore,
\begin{equation*}
  \dist(\M_G, \M_C)
  = \dist(r, \bar{r}) 
  = \frac{\lambda}{2m} {\textstyle \sum_{e \in E} \big( \dist(e,\bar{e}) + \dist(e',\bar{e}) \big)}
  = \frac{\lambda^2}{2m} \,.
\end{equation*}

($\Leftarrow$) Assume that there exists $\N = (N,\theta,\alpha) \in \MC[m+h+2]{}$
s.t. $\dist(\M_G,\N) \leq \frac{\lambda^2}{2m}$. We claim that $G$ has a vertex cover of size $h$.
Without loss of generality we may assume that the following hold for $\N$: 
\begin{enumerate}
\item $L(\N) \subseteq L(\M_G)$; \label{itm:labelsN}
\item $\N$ has initial state $\bar{r}$ with the same label than $r$, i.e., $\alpha(\bar{r}) = \ell(r)$; \label{itm:initialN}
\item $\ell(s) \in L(\N)$; and \label{itm:sink}
\item $\set{\ell(e)}{e \in E} \subseteq L(\N)$. \label{itm:alledges}
\item $\N$ has minimal size, i.e., $|N|$ is minimal \label{itm:sizeN}
\end{enumerate}
Assumption \eqref{itm:labelsN} follows by Lemma~\ref{lem:usefullabels}. 
Assumptions \eqref{itm:initialN}, \eqref{itm:sink} and \eqref{itm:alledges} are necessary to ensure $\dist(\M_G,\N) \leq \frac{\lambda^2}{2m}$. 
Indeed, if $\alpha(\bar{r}) \neq \ell(r)$, $\dist(\M_G,\N) = 1 > \frac{\lambda^2}{2m}$. If $\ell(s) \notin L(\N)$, then we get the following contradiction
\begin{align*}
 \dist(\M_G, \N) 
 &\geq |\Pr[{(\M_G)_\lambda}]{L^*\ell(s)L^\omega} - \Pr[{(\N)_\lambda}]{L^*\ell(s)L^\omega}|
 \tag{Corollary~\ref{cor:disctvupperbound}} \\
 &= \Pr[{(\M_G)_\lambda}]{L^*\ell(s)L^\omega} \tag{$\Pr[{(\N)_\lambda}]{L^*\ell(s)L^\omega} = 0$} \\
 &= \lambda^2 \left(1 - \frac{1}{m}\right) \tag{definition of $\Pr{}$} \\
 & > \frac{\lambda^2}{2m} \tag{$m \geq 2$} \,.
\end{align*}
Finally, if $\ell(e) \notin L(\N)$, for some $e \in E$, then we get the following contradiction
\begin{align*}
 \dist(\M_G, \N) &= \dist(r, \bar{r}) \tag{$r$ and $\bar{r}$ are the initial states} \\ 
 &\geq \lambda \cdot \tau(r)(\{e, e'\}) \tag{Lemma~\ref{lem:lowerbound}} \\
 &= \frac{\lambda}{m}  \tag{definition of $\M_G$}\\
 &> \frac{\lambda^2}{2m} \tag{$\lambda \in (0,1]$} \,.
\end{align*}

By assumption \eqref{itm:sizeN} we have that the sink state $s$ in $\M_G$ requires exacly one sink state $\bar{s}$ in $\N$ with same label than $s$ and having a self loop with probability $1$. Similarly, for vertex states $v$ in $\M$, $\N$ requires at most one vertex state $\bar{v}$ with same label than $v$ and having a self loop with probability $1$. Lastly, for each edge $e \in E$, $\N$ requires at most two edge states $\bar{e}$ and $\bar{e}'$ to respectively represent the edge states $e$ and $e'$ in $\M_G$; clearly $\ell(e) = \alpha(\bar{e})$ and $\ell(e') = \ell(\bar{e}')$.

By the assumptions made above, we have that $\{\bar{r},\bar{s} \} \subseteq N$ and 
$\{ \bar{e}, \bar{e}' \} \cap N \neq \emptyset$, for each edge $e \in E$. 
By Theorem~\ref{th:mincoupling}, there exists $\C \in \coupling{\M}{\N}$ such that $\dist = \discr{\C}$. Then,
\begin{align}
  \dist(\M_G, \N) &=  \dist(r, \bar{r}) \tag{$r$ and $\bar{r}$ are the initial states}\\
  &= \lambda \textstyle\sum_{m \in M, n \in N} \dist(m,n) \cdot \C(r,\bar{r})(m,n)  \tag{def. $\Gamma_\lambda^\C$ and $\discr{\C}$} \\
  &\geq \lambda \cdot  \min \left\{ \sum_{m \in M, n \in N} \dist(m,n) \cdot \omega(m,n) \, \left| \,  
  \begin{aligned}[c]
  & \omega \in \coupling{\tau(r)}{\mu} \\
  & \mu \in \D(N)
  \end{aligned}
  \right. \right\} \notag \\
  &= \lambda \sum_{e \in E} \min \left\{ f_e(\epsilon_e, \epsilon'_e) \, \left| \,  \epsilon_e, \epsilon'_e \in \left[0, \frac{1}{2m} \right] \right. \right\} \,,
  \label{eq:minfe}
\end{align}
where for arbitrary $e \in E$, $f_e$ is a function in $\epsilon_e$ and $\epsilon'_e$ defined as
\begin{equation*}
f_e(\epsilon_e, \epsilon'_e) = 
\left( \frac{1}{2m} - \epsilon_e \right) \cdot \dist(e,\bar{e}) +
\left( \frac{1}{2m} - \epsilon'_e \right) \cdot \dist(e',\bar{e}) +
\epsilon_e \cdot \dist(e,\bar{e}') + \epsilon'_e \cdot \dist(e',\bar{e}') \,.
\end{equation*}
The equality \eqref{eq:minfe} follows by the fact that an optimal coupling $\omega$ can be found among those with support included in $\bigcup_{e \in E} \{ (e,\bar{e}), (e',\bar{e}), (e,\bar{e}'),  (e',\bar{e}') \}$ (\cf\ Figure~\ref{fig:optcoupl} (left)). Note that the above formulation is general enough to model the case when $\bar{e} \notin N$ (resp. $\bar{e}' \notin N$) in which case $\epsilon_e = \epsilon'_e = \frac{1}{2m}$ (resp.\ $\epsilon_e = \epsilon'_e = 0$).
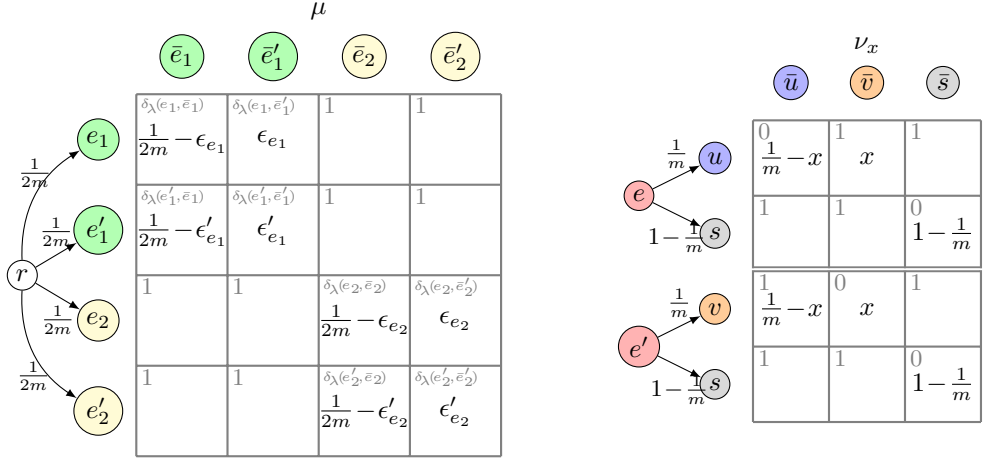
\begin{figure}
\centering
\begin{tikzpicture}[baseline=(current bounding box.center), labels]
\def\h{1.2cm}
\def\v{1.2cm}

\node at (1.2,-1.2) (corner) {};

 \draw[step=1.2cm,gray,thick] (1.2,-1.2) grid (1.2*5,-6);

\draw 
  ($(corner)+(-0.5,-\h*0.5)$) node[label,fill=green!30] (e1) {$e_1$}
  ($(e1)+(0,-\h)$) node[label,fill=green!30] (e1t) {$e'_1$}
  ($(e1)+(0,-2*\h)$) node[label,fill=yellow!20] (e2) {$e_2$}
  ($(e1)+(0,-3*\h)$) node[label,fill=yellow!20] (e2t) {$e'_2$}
  ($(e1)!0.5!(e2t)+(-1,0)$) node[label] (r) {$r$}
  
  ($(corner)+(0.5*\v,0.5)$) node[label,fill=green!30] (e1b) {$\bar{e}_1$}
  ($(e1b)+(\v,0)$) node[label,fill=green!30] (e1tb) {$\bar{e}'_1$}
  ($(e1b)+(2*\v,0)$) node[label,fill=yellow!20] (e2b) {$\bar{e}_2$}
  ($(e1b)+(3*\v,0)$) node[label,fill=yellow!20] (e2tb) {$\bar{e}'_2$}
  ($(e1b)!0.5!(e2tb)+(0,0.6)$) node (rb) {$\mu$}
  
  ($(corner)!0.5!($(corner)+(\v,-\h)$)$) node (cellw) {$\frac{1}{2m}\!-\! \epsilon_{e_1}$}
  ($(cellw)+(0,-\h)$) node {$\frac{1}{2m}\!-\! \epsilon'_{e_1}$}
  ($(cellw)+(\v,0)$) node {$\epsilon_{e_1}$}
  ($(cellw)+(\v,-\h)$) node {$\epsilon'_{e_1}$}
  
  ($($(corner)+(2*\v,-2*\h)$)!0.5!($(corner)+(3*\v,-3*\h)$)$) node (cellw1) {$\frac{1}{2m}\!-\! \epsilon_{e_2}$}
  ($(cellw1)+(0,-\v)$) node {$\frac{1}{2m}\!-\! \epsilon'_{e_2}$}
  ($(cellw1)+(\h,0)$) node {$\epsilon_{e_2}$}
  ($(cellw1)+(\h,-\v)$) node {$\epsilon'_{e_2}$}
  
  ($(corner)+(0.4*\h,-0.15)$) node[font=\tiny, color=gray] (cellwcost) {$\dist[\!\lambda]\!(\!e_1\!,\!\bar{e}_1\!)$}
  ($(cellwcost)+(\h,0)$) node[font=\tiny, color=gray]  {$\dist[\!\lambda]\!(\!e_1\!,\!\bar{e}'_1\!)$}
  ($(cellwcost)+(0,-\v)$) node[font=\tiny, color=gray]  {$\dist[\!\lambda]\!(\!e'_1\!,\!\bar{e}_1\!)$}
  ($(cellwcost)+(\h,-\v)$) node[font=\tiny, color=gray]  {$\dist[\!\lambda]\!(\!e'_1\!,\!\bar{e}'_1\!)$}
  
  ($($(corner)+(2*\h,-2*\v)$)+(0.4*\h,-0.15)$) node[font=\tiny, color=gray] (cellwcost1) {$\dist[\!\lambda]\!(\!e_2\!,\!\bar{e}_2\!)$}
  ($(cellwcost1)+(\h,0)$) node[font=\tiny, color=gray]  {$\dist[\!\lambda]\!(\!e_2\!,\!\bar{e}'_2\!)$}
  ($(cellwcost1)+(0,-\v)$) node[font=\tiny, color=gray]  {$\dist[\!\lambda]\!(\!e'_2\!,\!\bar{e}_2\!)$}
  ($(cellwcost1)+(\h,-\v)$) node[font=\tiny, color=gray]  {$\dist[\!\lambda]\!(\!e'_2\!,\!\bar{e}'_2\!)$}
  
  ($($(corner)+(2*\h,0)$)+(0.12*\h,-0.15)$) node[font=\small, color=gray] (cellwcost2) {$1$}
  ($(cellwcost2)+(\h,0)$) node[font=\small, color=gray]  {$1$}
  ($(cellwcost2)+(0,-\v)$) node[font=\small, color=gray]  {$1$}
  ($(cellwcost2)+(\h,-\v)$) node[font=\small, color=gray]  {$1$}
  
  ($($(corner)+(0,-2*\v)$)+(0.12*\h,-0.15)$) node[font=\small, color=gray] (cellwcost2) {$1$}
  ($(cellwcost2)+(\h,0)$) node[font=\small, color=gray]  {$1$}
  ($(cellwcost2)+(0,-\v)$) node[font=\small, color=gray]  {$1$}
  ($(cellwcost2)+(\h,-\v)$) node[font=\small, color=gray]  {$1$}
  ;
  
  \path[-latex, font=\small]
  (r) edge[bend left] node[above] {$\frac{1}{2m}$} (e1)
  	edge node[above] {$\frac{1}{2m}$} (e1t)
	edge node[below] {$\frac{1}{2m}$} (e2)
	edge[bend right] node[below] {$\frac{1}{2m}$} (e2t)
  ;
\end{tikzpicture}
\qquad\qquad
\begin{tikzpicture}[baseline=(current bounding box.center), labels]
\def\h{1cm}
\def\v{1cm}
 \draw[step=1cm,gray,thick] (1,-1) grid (4,-5);
 \draw[gray,thick] (1,-2.95) -- (4,-2.95);
 \draw[fill=white, draw=white] (0.5,-2.97) rectangle (4.3,-2.98);
\draw 
  (-0.5,-2) node[label,fill=red!30] (e) {$e$}
  ($(e)+(\h,0.5*\v)$) node[label,fill=blue!30] (u) {$u$}
  ($(e)+(\h,-0.5*\v)$) node[label,fill=gray!30] (s0) {$s$}
  
  ($(e)+(0,-2*\v)$) node[label,fill=red!30] (et) {$e'$}
  ($(et)+(\h,0.5*\v)$) node[label,fill=orange!40] (v) {$v$}
  ($(et)+(\h,-0.5*\v)$) node[label,fill=gray!30] (s1) {$s$}  
  
  ($(u)+(\v,\v)$) node[label,fill=blue!30] (ub) {$\bar{u}$}
  ($(ub)+(\v,0)$) node[label,fill=orange!40] (vb) {$\bar{v}$}
  ($(vb)+(\v,0)$) node[label,fill=gray!30] (sb) {$\bar{s}$}
  ($(sb)!0.5!(ub)+(0,0.5)$) node (eb) {$\nu_x$}
  
  (1.5,-1.5) node {$\frac{1}{m} \!-\!x$}
  (1.15,-1.15) node[font=\small, color=gray] {$0$}
  (2.5,-1.5) node {$x$}
  (2.15,-1.15) node[font=\small, color=gray] {$1$}
  (3.15,-1.15) node[font=\small, color=gray] {$1$}
  (1.15,-2.15) node[font=\small, color=gray] {$1$}
  (2.15,-2.15) node[font=\small, color=gray] {$1$}
  (3.5,-2.5) node {$1\!-\!\frac{1}{m}$}
  (3.15,-2.15) node[font=\small, color=gray] {$0$}
  
  (1.5,-3.5) node {$\frac{1}{m} \!-\!x$}
  (1.15,-3.15) node[font=\small, color=gray] {$1$}
  (2.5,-3.5) node {$x$}
  (2.15,-3.15) node[font=\small, color=gray] {$0$}
  (3.15,-3.15) node[font=\small, color=gray] {$1$}
  (1.15,-4.15) node[font=\small, color=gray] {$1$}
  (2.15,-4.15) node[font=\small, color=gray] {$1$}
  (3.5,-4.5) node {$1\!-\!\frac{1}{m}$}
  (3.15,-4.15) node[font=\small, color=gray] {$0$}
  ;
  
  \path[-latex, font=\small]
  (e) edge node[above] {$\frac{1}{m}$} (u)
  (e) edge node[below] {$1\!-\! \frac{1}{m}$} (s0)
  
  (et) edge node[above] {$\frac{1}{m}$} (v)
  (et) edge node[below] {$1\!-\! \frac{1}{m}$} (s1)  
  
  ;
\end{tikzpicture}
\caption{(Left) Couplings $\omega \in \coupling{\tau(r)}{\mu}$ parametric in $\epsilon, \epsilon' \in [0,\frac{1}{2m}]^m$, for $m =2$; 
(Right) Couplings $\omega \in \coupling{\tau(e)}{\nu_x}$ and $\tilde{\omega} \in \coupling{\tau(e')}{\nu_x}$ parametric in $x \in [0,\frac{1}{m}]$ (\cf\ Theorem~\ref{th:NPhardAM}).}
\label{fig:optcoupl}
\end{figure}

Consider an arbitrary edge $e \in E$ with endpoints $u,v \in V$. Next, we lower-bounds for $\min \left\{ f_e(\epsilon_e, \epsilon'_e) \, \left| \,  \epsilon_e, \epsilon'_e \in \left[0, \frac{1}{2m} \right] \right. \right\}$. 
The edge $e$ may be covered by $\bar{V}$ (i.e., $\bar{u}$ or $\bar{v}$ appear in $\N$, \cf\ Case~\ref{itm:coverrededge}) or not (i.e., neither $\bar{u}$ nor $\bar{v}$ appear in $\N$, \cf\ Case~\ref{itm:uncoverrededge}).
\begin{enumerate}[label={\Alph*.}]
\item \label{itm:uncoverrededge} If neither $\bar{u}$ nor $\bar{v}$ appear in $\N$. Then, by Lemma~\ref{lem:lowerbound}, for arbitrary $p \in \{e,e'\}$ and $q \in \{\bar{e}, \bar{e}'\}$, we have $\dist(p,q) \geq \frac{\lambda}{m}$.  Therefore, for arbitrary $\epsilon_e, \epsilon'_e \in [0, \frac{1}{2m}]$
\begin{equation*}
f_e(\epsilon_e, \epsilon'_e) \geq \frac{\lambda}{m} \left( \frac{1}{2m} - \epsilon_e + \frac{1}{2m} - \epsilon'_e + \epsilon_e + \epsilon'_e \right) = \frac{\lambda}{m^2} \,.
\end{equation*}

\item \label{itm:coverrededge} If $\bar{u}$ or $\bar{v}$ appear in $\N$. We distinguish two subcases: when both $\bar{e}$ and $\bar{e}'$ appear in $\N$, or only one of them.
\begin{enumerate}[label*={\arabic*.}]
\item If $\bar{e}' \notin N$ (resp.\ $\bar{e} \notin N$), then $\epsilon_e = \epsilon'_e = 0$ 
(resp.\ $\epsilon_e = \epsilon'_e = \frac{1}{2m}$). 
\begin{align*}
&f_e(\epsilon_e, \epsilon'_e) \\
&= \frac{1}{2m} \big( \dist(e,\bar{e}) + \dist(e',\bar{e}) \big) 
\tag{$\epsilon_e = \epsilon'_e = 0$} \\
&= \frac{\lambda}{2m} \sum_{p \in M, q \in N} \big( \C(e,\bar{e})(p,q) + \C(e',\bar{e})(p,q) \big) \cdot \dist(p,q) \tag{def.\ $\C$ and $\discr{\C}$}\\
&\geq \frac{\lambda}{2m} \cdot \min \left\{ \sum_{p \in M,q \in N} \big(\omega(p,q) + \tilde{\omega}(p,q)\big) \cdot \dist(p,q)
\, \left| \,
\begin{aligned}[c]
	& \omega \in \coupling{\tau(e)}{\nu}\\
	& \tilde{\omega} \in \coupling{\tau(e')}{\nu}\\
	& \nu \in \D(N)
\end{aligned}
\right. \right\} \\
&= \frac{\lambda}{2m} \cdot \min \left\{\left. x + \left( \frac{1}{m} -x \right)\right| 0\leq x \leq 1-\frac{1}{m} \right\} \tag{\cf~Figure.~\ref{fig:optcoupl} (right)}\\
&= \frac{\lambda}{2m^2} \,.
\end{align*} 
Analogously, the case $\bar{e} \notin N$ has the same lower-bound. 
\item If both $\bar{e}$ and $\bar{e}'$ are in $\N$, then $0 < \epsilon_e + \epsilon'_e < \frac{1}{m}$. We further spit in two subcases: when both $\bar{u}$ and $\bar{v}$ appear in $\N$, or only one of them.
\begin{enumerate}[label*={\arabic*.}]
\item If both $\bar{u}$ and $\bar{v}$ are in $\N$, then $f_e(\epsilon_e, \epsilon'_e) \geq 0$ because
$\dist(p,q) \geq 0$ for arbitrary $p \in M$ and $q \in N$. 
\item If $\bar{u} \notin N$ (resp.\ $\bar{v} \notin N$), then, for any $\epsilon_e, \epsilon'_e \in [0, \frac{1}{m}]$
\begin{align*}
&f_e(\epsilon_e, \epsilon'_e) \\ 
&\geq \left( \frac{1}{2m} - \epsilon_e \right) \cdot \dist(e,\bar{e}) +
\epsilon_e \cdot \dist(e,\bar{e}') \tag{$\dist(e',q) \geq 0$ for any $q \in N$} \\
&\geq \lambda \cdot \tau(e)(u) \cdot \left( \frac{1}{2m} - \epsilon_e + \epsilon_e \right) \tag{Lemma~\ref{lem:lowerbound}} \\
&= \frac{\lambda}{2m^2} \tag{$\tau(e)(u) = \frac{1}{m}$} \,.
\end{align*}
Similarly, the case $\bar{v} \notin N$ has the same lower-bound.
\end{enumerate}
\end{enumerate}
\end{enumerate}

By assigning with each edge $e \in E$ any among the cases A, B.1, B.2.1, and B.2.2 described above, we induce a selection of the states of $\N$ and, at the same time, a selection of vertices  $\bar{V} \subseteq V$ in the graph $G$, namely, $\bar{V} = \set{v \in V}{\bar{v} \in N}$. Recall that at least $m$ states in $\N$ have to be edge vertices and two other states are reserved respectively for the sink $\bar{s}$ and the initial state $\bar{r}$, therefore $|\bar{V}| \leq h$.

Assume that the $m$ edges in $G$ are assigned the above case as follows:
\begin{equation*}
  m = \overbrace{m'}^{\text{B.2.1}} + \overbrace{m''}^{\text{A}} + \overbrace{m'''}^{\text{B.1 or B.2.2}} \,.
\end{equation*}
In the above assignment $m' + m'''$ edges are covered by $\bar{V}$ whereas $m''$ edges are not covered by $\bar{V}$. Necessarily, $m'' \leq m'$. Otherwise, we have that
\begin{align*}
 \dist(\M,\N) &\geq \lambda \sum_{e \in E} \min \left\{ f_e(\epsilon_e, \epsilon'_e) \, \left| \,  \epsilon_e, \epsilon'_e \in \left[0, \frac{1}{2m} \right] \right. \right\} \tag{by Equation~\eqref{eq:minfe}}\\
 &\geq \lambda \left( m'' \cdot \frac{\lambda}{m^2}  + m''' \cdot \frac{\lambda}{2m^2}  \right) \tag{lower-bounds proved before}\\
 &=  \lambda  \left( m''  \cdot \frac{\lambda}{m^2}  + (m - m' - m'') \cdot  \frac{\lambda}{2m^2}  \right) \tag{$m''' = m - m' - m''$} \\
 &=  \frac{\lambda^2}{2m^2} \cdot (m - m' + m'') \\
 &> \frac{\lambda^2}{2m} \tag{$m'' > m'$}
 \end{align*}
which contradicts the assumption that $\dist(\M,\N) \leq \frac{\lambda^2}{2m}$.

If $m'' = 0$, then $\bar{V}$ is a vertex cover for $G$ an we are done. 
\begin{figure}[t]
\centering
\begin{tikzpicture}[baseline=(current bounding box.center), labels]
\draw 
  (0,0) node[label,label,fill=red!30] (e) {$\bar{e}$}
  (e)+(right: 1.5) node[label,label,fill=red!30] (e1) {$\bar{e}'$}
  
  (e)+(down: 1.5) node[label,fill=orange!40] (u) {$\bar u$}
  (e1)+(down: 1.5) node[label,fill=blue!20!white,dashed] (v) {$\bar v$}
  ;
\path[-latex, font=\small,dashed]
 (e) edge (u) edge (v)
 (e1) edge (u) edge (v);
\end{tikzpicture}
\quad$\Longrightarrow$\quad
\begin{tikzpicture}[baseline=(current bounding box.center), labels]
\draw 
  (0,0) node[label,label,fill=red!30] (e) {$\bar{e}$}
  
  (e)+(-0.8,-1.5) node[label,fill=orange!40] (u) {$\bar u$}
  (e)+(0.8,-1.5) node[label,fill=blue!20!white,dashed] (v) {$\bar v$}
  ;
\path[-latex, font=\small,dashed]
 (e) edge (u) edge (v);
\end{tikzpicture}
\qquad\qquad\quad
\begin{tikzpicture}[baseline=(current bounding box.center), labels]
\draw 
  (0,0) node[label,label,fill=red!30] (e) {$\bar{e}$}
  (e)+(right: 1.5) node[label,label,fill=red!20,dashed] (e1) {$\bar{e}'$}
  (e)+(down: 1.5) node (u) {}
  ;
\end{tikzpicture}
\quad$\Longrightarrow$\quad
\begin{tikzpicture}[baseline=(current bounding box.center), labels]
\draw 
  (0,0) node[label,label,fill=red!30] (e) {$\bar{e}$}
  (e)+(right: 1.5) node[label,label,fill=red!20,dashed] (e1) {$\bar{e}'$}
  (e)+(down: 1.5) node[label,fill=orange!40] (u) {$\bar u$}
  ;
\path[-latex, font=\small,dashed]
 (e) edge (u)
 (e1) edge (u);
\end{tikzpicture}
\caption{(Left) By removing one of the edge states we turn the assignment from type B.2.1 to type B.1 or B.2.1; (Right) By adding one state vertex corresponding to one endpoint of the edge we turn the assignment from type A to type B.1 or B.2.1.}
\label{fig:edgeassign}
\end{figure}
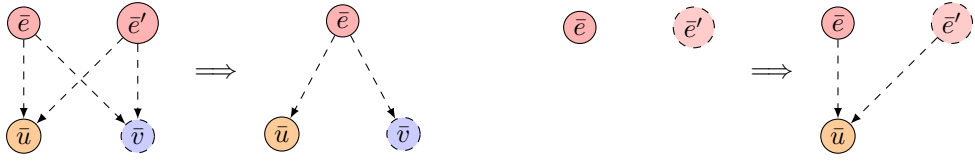
Otherwise, we claim that there is another assignment of the the edges where no edge is of type A inducing a vertex cover.
Such assignment is obtained from the previous one, leveraging from the fact that $m'' \leq m'$. We proceed by turning all $m''$ edges of type A to edges of type B.1 or B.2.2 removing exactly one of the two edge states from (at most $m''$) edges of type B.2.1 (\cf\ Fig.~\ref{fig:edgeassign}~(left)). This gives enough room to add vertex states in a way that all edges of type A are turned into edges of type B.1 or B.2.2 (\cf\ Fig.~\ref{fig:edgeassign}~(right)).

This concludes the proof of the reduction from \VC. \qed
\end{proof}

\section{Minimum Significant Approximant Bound}
\label{sec:MSAB}

Recall that, having two MCs that are at distance $1$ from each other means that there is no significant similarity between their behaviors. Accordingly, we say that an MC $\N$ is a \emph{significant approximant} for 
the MC $\M$ w.r.t. $\dist$ if, and only if, $\dist(\M,\N) < 1$.

The \emph{Minimum Significant Approximant Bound} problem w.r.t.\ $\dist$ ($\MSAB[\lambda]$)
looks for the smallest positive integer $k$ such that $\N \in \MC[k]{}$ is a significant approximant 
for a given an MC $\M$.
The decision version of this problem is called 
\emph{Significant Bounded Approximant problem} w.r.t.\ $\dist$ ($\SBA[\lambda]$), and asks whether,
for a given positive integer $k$, there exists $\N \in \MC[k]{}$ such that $\dist(\M,\N) < 1$.

When the distance $\dist$ is discounted (i.e., $\lambda < 1$), the two problems above turn out to be trivial because for any MCs $\M$, $\N$ with initial states labelled with same label, $\dist(\M,\N) \leq \lambda$, thus
the minimum size for a significant approximant is always $1$.
In contrast, we show that when the distance $\dist$ is undiscounted ($\lambda = 1$) the same two
problems are NP-complete. The NP-completeness result is obtained via a characterization of $\SBA$ as a combinatorial problem in graph theory on vertex-labelled directer graphs.

A \emph{vertex-labelled directed graph} is a directed graph $G = (V, E)$ with a \emph{vertex labelling function} associating with each vertex in $V$ a label. 
For a a Markov chain $\M = (M,\tau,\ell)$, we denote by $G(\M)$ its underlying vertex-labelled directed graph, having $M$ as set of vertices labelled by $\ell$ and directed edges $(m,m')$ if and only if $\tau(m)(m') > 0$.

\begin{definition}[Reflected paths]
A path $v_0 \dots v_n$ in a vertex-labeled directed graph $G$, is \emph{reflected} in a subgraph $G'$ of $G$ if there exists a path $v'_0 \dots v'_n$ in $G'$ such that $v_n = v'_n$ and, for all $0 \leq i \leq n$, $v_i$ and $v'_i$ have the same label.

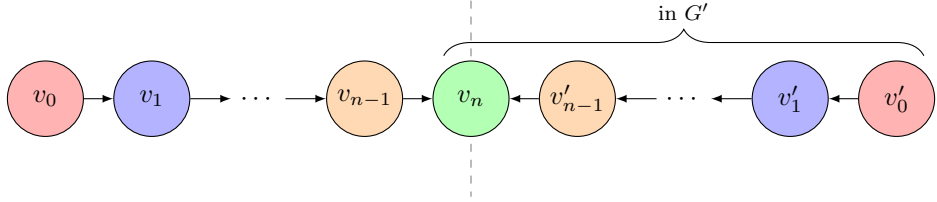
\begin{center}
\begin{tikzpicture}[baseline=(current bounding box.center), labels]
\def\h{1.4cm}

\draw (4*\h,0) coordinate (vn);
\draw[dashed, color=gray] ($(vn)+(up:1.3cm)$) -- ($(vn)+(down:1.3cm)$);

\draw 
  (0,0) node[label,minimum width=1cm,fill=red!30] (v) {$v_0$}
  ($(v)+(right:\h)$) node[label,minimum width=1cm,fill=blue!30] (v1) {$v_1$}
  ($(v)+(right:2*\h)$) node (vi) {$\cdots$}
  ($(v)+(right:3*\h)$) node[label,minimum width=1cm,fill=orange!30] (vn1) {$v_{n-1}$}
  ($(v)+(right:4*\h)$) node[label,minimum width=1cm,fill=green!30] (vn) {$v_n$}
  ($(v)+(right:5*\h)$) node[label,minimum width=1cm,fill=orange!30] (vpn1) {$v'_{n-1}$}
  ($(v)+(right:6*\h)$) node (vpi) {$\cdots$}
  ($(v)+(right:7*\h)$) node[label,minimum width=1cm,fill=blue!30] (vp1) {$v'_1$}
  ($(v)+(right:8*\h)$) node[label,minimum width=1cm,fill=red!30] (vp) {$v'_0$}
  ;
  
  \draw [decorate,decoration={brace,amplitude=10pt,raise=6pt},yshift=0pt]
	(vn.north west) -- (vp.north east) node [midway,yshift=0.8cm] {\footnotesize in $G'$};

\path[-latex, font=\small]
     (v) edge (v1)
     (v1) edge (vi)
     (vi) edge (vn1) 
     (vn1) edge (vn)
     (vp) edge (vp1)
     (vp1) edge (vpi)
     (vpi) edge (vpn1)
     (vpn1) edge (vn)
;
\end{tikzpicture}
\end{center}

\end{definition}

We show that deciding $\instance{\M, k} \in \SBA$ is equivalent to checking whether the underlying graph of $\M$ has a bottom strongly connected component (BSCC) and a path reaching it from the initial state satisfying the following condition: the size of the BSCC plus the number of labels occurring in a prefix of such path, obtained by removing the part that is reflected in the BSCC, does not exceed $k$.
\begin{lemma} \label{lem:SAasFindingSCCinGraph}
Let $\M$ be a minimal MC with initial state $m_0$.  
Then, $\instance{\M, k} \in \SBA$ iff  $G(\M)$ has a BSCC $G' = (V,E)$ and a path $m_0 \dots m_h$, such that, for some $0 \leq p \leq h$, 
\begin{enumerate}
\item the path $m_p \dots m_h$ is reflected in $G'$, and 
\item $|\set{\ell(m_i)}{ i < p}| + |V| \leq k$.
\end{enumerate}

\end{lemma}
\begin{proof}
($\Rightarrow$) By hypothesis there exists $\N = (N, \theta, \alpha) \in \MC{}$ such that $\dist[1](\M,\N) < 1$. We assume that $\N$ is minimal (otherwise one can replace it with its bisimilarity quotient). By Lemma~\ref{lem:dualdiscrepancy} and Theorem~\ref{th:mincoupling}, there exists $\C \in \coupling{\M}{\N}$ such that $\beta^\C_1(\M,\N) > 0$. 
Therefore there exists a path $(m_0,n_0)\dots (m_p,n_p)$ in $G(\C)$, such that $\ell(m_i) = \alpha(n_i)$ (for $i = 0..p$) and $\gamma^\C_1(m_p, n_p) = 0$.

Note that for arbitrary $m \in M$ and $n \in N$ such that $\gamma^\C_1(m,n) = 0$, the following hold
\begin{equation*}
0 = \gamma^\C_1(m,n) = \textstyle \sum_{u \in M} \sum_{v \in N} \gamma^\C_1(u,v) \cdot \C(m,n)(u,v) \,.
\end{equation*}
Therefore, for any $u \in M$ and $v \in N$ we have that $\C(m,n)(u,v) > 0$ implies $\discr[1]{\C}(u,v) = 0$. 

Let $R \subseteq M \times N$ be the set of states reachable from $(m_p, n_p)$ in $G(\C_1)$. 
By $\gamma^\C_1(m_p, n_p) = 0$ and what have been said before we have that $(m,n) \in R$ implies that $m \sim n$. Let $G = (V,E)$ be a bottom strongly connected component of $G(\C)$ such that $V \subseteq R$.

Consider now the graphs $G_1 = (V_1, E_1)$ and $G_2 = (V_2, E_2)$ where
\begin{align}
&V_1 = \{ m \mid (m,n) \in V \} && E_1 = \{ (m,u) \mid \tau(m)(u) > 0 \text{ and } m,u \in V_1 \}, \\
&V_2 = \{ n \mid (m,n) \in V \} && E_2 = \{ (n,v) \mid \theta(n)(v) > 0  \text{ and }  n,v \in V_1 \}.
\end{align}
Since $\M$ and $\N$ are minimal and $m \sim n$ for all $(m,n) \in V$ we have that for all $(m,n), (u,v) \in V$, $C(m,n)(u,v) = \tau(m)(u) = \theta(n)(v)$. Therefore $G_1$ and $G_2$ are bottom strongly connected components of $G(\M)$ and $G(\N)$ respectively, and are isomorphic with each other.

Let now take the path $(m_0,n_0)\dots (m_{h}, n_{h})$ in $G(\C_1)$ such that $(m_{h}, n_{h}) \in V$ obtained by appending the path $(m_0,n_0)\dots (m_{p}, n_{p})$ with the path $(m_{p}, n_{p}) \dots (m_{h}, n_{h})$. Note that such a path exists since $(m_{h}, n_{h}) \in R$ and, $\ell(m_i) = \alpha(n_i)$ for all $0 \leq i \leq h$. 

There are two possible cases:
\begin{trivlist}
\item Case 1: if $n_i \notin V_2$ for all $0 \leq i < h$ we have that the 
zero-step path $m_h$ is trivially reflected in $G_1$ and 
\begin{equation*}
|\set{\ell(m_i)}{ i < h} | + |V_1| = |\set{\ell(n_i)}{ i < h} | + |V_2| \leq |\set{n_i}{ i < h} | + |V_2| \leq |N| \leq k \,.
\end{equation*}
\item Case 2: If $n_i \in V_2$ for some $0 \leq i < h$. Let $q < h$ be the smallest index such that $n_q \in V_2$. Since $C(m_q, n_q)(n_{q+1}, n_{q+1}) > 0$ implies $\theta(n_q)(n_{q+1}) > 0$ and $G_2$ is a bottom strongly connected component, we have that also $n_{q+1} \in V_2$. This shows that $n_q \dots n_h$ is a path in $G_2$. Since the isomorphism between $G_2$ and $G_1$ preserves the labels 
(indeed, any $n \in V_2$ is mapped with the unique state $m \in V_1$ such that $m \sim n$) we can see that there exists a path $v_p \dots v_{h-1}m_h$ in $G_1$ such that $\ell(m_i) = \ell(v_i)$ for all $p \leq i < h$. Therefore we have that the path $m_p \dots m_h$ is reflected in $G_1$ and 
\begin{gather*}
|\set{\ell(m_i)}{ i < p }| + |V_1| = |\set{\ell(n_i)}{ i < p }| + |V_2| \leq |\set{n_i}{ i < p }| + |V_2| \leq |N| \leq k \,.
\end{gather*}
\end{trivlist}

($\Leftarrow$) Let $M'= \{m_0, \dots, m_h\}$. Assume w.l.o.g.\ that $M' \cap V = \{m_h\}$ (otherwise one can consider a prefix of the path that enjoys the assumption). By hypothesis we have that $m_p \dots m_h$ is reflected in $G'$, therefore there exists a path $v_p \dots v_h$ in $G'$ with $v_h = m_h$ and $\ell(m_j) = \ell(v_j)$ for all $p \leq j \leq h$. 

Let $\{l_0, \dots, l_q \} = \set{\ell(m_i)}{i < p}$ and $N' = \{n_0, \dots, n_{q}\}$. Consider the Markov chain $\N = (N, \theta, \alpha)$ where $N = N' \cup V$, and 
\begin{align*}
&\theta(n) = \begin{cases}
\tau(n) &\text{if $n \in V$} \\
\sum_{v \in N} (1/|N|) \cdot 1_v &\text{if $n \in N'$}
\end{cases}
&&\alpha(n) = \begin{cases}
\ell(n) &\text{if $n \in V$} \\
l_i &\text{if $n = n_i$ for some $0\leq i \leq q$}
\end{cases}
\end{align*}
Note that $\theta$ is well defined because the support of $\tau(v)$ is included in $V$ for all $v \in V$. 

By construction, states in $N'$ have pairwise distinct labels, therefore we can define the function $f \colon M' \to N$ as $f(m_i) = n$ if $0 \leq i < p$ and $n \in N'$ such that $\alpha(n) = \ell(m_i)$; and $f(m_i) = v_i$ if $p \leq i \leq h$.
In the following we will prove that for all $m_i \in M'$, $\dist[1](m_i,f(m_i)) < 1$. 
We proceed by induction on $r = h-i$.
\begin{trivlist}
\item \textsc{Base Case} ($i = h$): One can readily show that $\dist[1](m_h, m_h) = 0$.
\item \textsc{Inductive Step} ($i < h$): 
Let $n = f(m_{i})$ and $n' = f(m_{i+1})$ then the following hold
\begin{align}
&\dist[1](m_{i}, f(m_{i})) = \mathcal{K}(\dist[1])(\tau(m_i),\theta(n)) \tag{$\ell(m_{i}) = \alpha(n)$} \\
& = \min \textstyle \set{ \sum_{u \in M} \sum_{v \in N} \omega(u,v) \cdot \dist[1](u, v) }{ \omega \in \coupling{\tau(m_{i})}{\theta(n)}} \tag{def. $\mathcal{K}$ \cf~\eqref{eq:Kantorivichdef}} \\
& \leq \textstyle \sum_{u \in M} \sum_{v \in N} \tau(m_{i})(u) \cdot \theta(n)(v) \cdot \dist[1](u, v) 
 \tag{*}\label{eq:parsync} \\ 
&\leq \tau(m_i)(m_{i+1})\cdot \theta(n)(n') \cdot \dist[1](m_{i+1}, n') + (1 - \tau(m_i)(m_{i+1})\cdot \theta(n)(n')) \tag{$\dist[1] \sqsubseteq \mathbf{1}$ } \\
&< 1 \tag{$\dist[1](m_{i+1}, f(m_{i+1})) < 1$ and $\tau(m_i)(m_{i+1})\cdot \theta(n)(n') > 0$ }
\end{align}
Inequality~\eqref{eq:parsync} is due to the fact that given $\mu, \nu \in \D(X)$, their `parallel synchronisation' $\mu \odot \nu \in \D(X \times X)$ defined as $(\mu \odot \nu)(x,y) = \mu(x) \cdot \nu(y)$ is a coupling for $(\mu, \nu)$. 
\end{trivlist}

Therefore $\dist[1](\M,\N) = \dist[1](m_0, f(m_0)) < 1$. By construction $|N| \leq k$, therefore $\instance{\M, k} \in \SBA$. \qed
\end{proof}

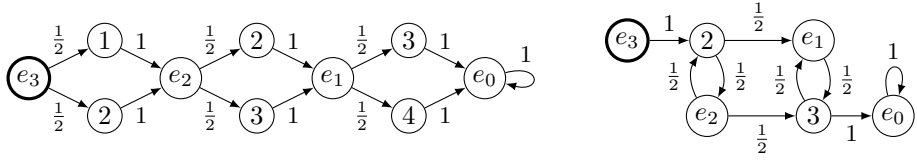
\begin{figure}[t]
\centering
\begin{tikzpicture}[baseline=(current bounding box.center), labels]
\def\h{2cm}
\def\v{0.5cm}
\draw 
  (0,0) node[initial] (a) {$e_3$}
  ($(a)+(right:\h)$) node[label] (b) {$e_2$}
  ($(b)+(right:\h)$) node[label] (c) {$e_1$}
  ($(c)+(right:\h)$) node[label] (s) {$e_0$}
  
  ($(a)!0.5!(b)+(up:\v)$) node[label] (a1) {$1$}
  ($(a)!0.5!(b)+(down:\v)$) node[label] (a2) {$2$}
  
  ($(b)!0.5!(c)+(up:\v)$) node[label] (b1) {$2$}
  ($(b)!0.5!(c)+(down:\v)$) node[label] (b2) {$3$}
  
  ($(c)!0.5!(s)+(up:\v)$) node[label] (c1) {$3$}
  ($(c)!0.5!(s)+(down:\v)$) node[label] (c2) {$4$}
  ;
\path[-latex, font=\small]
  (a) edge node[near start, above] {$\frac{1}{2}$} (a1)
        edge node[near start, below] {$\frac{1}{2}$} (a2)
  (b) edge node[near start, above] {$\frac{1}{2}$} (b1)
        edge node[near start, below] {$\frac{1}{2}$} (b2)
  (c) edge node[near start, above] {$\frac{1}{2}$} (c1)
        edge node[near start, below] {$\frac{1}{2}$} (c2)
        
  (a1) edge node[above] {$1$} (b)
  (a2) edge node[below] {$1$} (b)
  (b1) edge node[above] {$1$} (c)
  (b2) edge node[below] {$1$} (c)
  (c1) edge node[above] {$1$} (s)
  (c2) edge node[below] {$1$} (s)
  
  (s) edge[loop right] node[near start, above] {$1$} (s)
;
\end{tikzpicture}
\qquad
\begin{tikzpicture}[baseline=(current bounding box.center), labels]
\def\h{3.5cm}
\def\v{0.5cm}
\draw 
  (0,0) coordinate (init)
  ($(a)+(right:\h)$) coordinate (end)
  
  ($(init)+(up:\v)$) node[initial] (a) {$e_3$}
  ($(init)!0.3!(end)+(up:\v)$) node[label] (2) {$2$}
  ($(init)!0.3!(end)+(down:\v)$) node[label] (b) {$e_2$}
  ($(init)!0.7!(end)+(up:\v)$) node[label] (c) {$e_1$}
  ($(init)!0.7!(end)+(down:\v)$) node[label] (3) {$3$}
  ($(end)+(down:\v)$) node[label] (s) {$e_0$}
  ;
\path[-latex, font=\small]
  (a) edge node[above] {$1$} (2)
  (2) edge[bend left] node[right] {$\frac{1}{2}$} (b)
        edge node[above] {$\frac{1}{2}$} (c)
  (b) edge[bend left] node[left] {$\frac{1}{2}$} (2)
  (b) edge node[below] {$\frac{1}{2}$} (3)
  (3) edge[bend left] node[left] {$\frac{1}{2}$} (c)
        edge node[below] {$1$} (s)
  (c) edge[bend left] node[right] {$\frac{1}{2}$} (3)
  (s) edge[loop above] node[above] {$1$} (s)
  ;
\end{tikzpicture}
\caption{(Left) The MC $\M_G$ associated to the graph $G$ in Figure~\ref{fig:ReductionFromVertexCover} and
(right) an MC $\N$ associated to the vertex cover $C = \{2,3\}$ of $G$ such that $\dist[1](\M_G,\N) < 1$ 
(cf.\ Theorem~\ref{th:NP-completeSA}).}
\label{fig:ReductionFromMinVertexCover}
\end{figure}

\begin{theorem} \label{th:NP-completeSA}
$\SBA$ is NP-complete. 
\end{theorem}
\begin{proof}
The membership in NP is easily proved by using the characterization in Lemma~\ref{lem:SAasFindingSCCinGraph} and exploiting Tarjan's algorithm for generating bottom SCCs.
As for the NP-hardness, we provide a polynomial-time many-one reduction from \VC. 
Let $G = (V,E)$ be a graph with $E = \{e_1, \dots, e_n \}$. 
We construct the MC $\M_G$ as follows. The set of states is given by the set of edges $E$ along with
two states $e^1_i$ and $e^2_i$, for each edge $e_i \in E$, representing the two endpoints of $e_i$ and 
an extra sink state $e_0$. The initial state is $e_n$. The transition probabilities are given as follows.
The sink state $e_0$ loops with probability $1$ to itself. Each edge $e_i \in E$ 
goes with probability $\frac{1}{2}$ to $e^1_i$ and $e^2_i$, respectively. 
For $1 \leq i \leq n$, the states $e^1_i$ and $e^2_i$ go with probability $1$ to the state $e_{i-1}$.
The edge states and the sink state
are labelled by pairwise distinct labels, while the endpoints states $e^1_i$ and $e^2_i$ are labelled by
the node in $V$ they represent.  
An example of construction for $\M_G$ is shown in Figure~\ref{fig:ReductionFromMinVertexCover}.

Next we show the following equivalence:
\begin{align}
   \instance{G,h} \in \VC
  &&
  \text{iff}
  &&
  \instance{\M_G,h + n + 1} \in \SBA
  \label{eq:SAreduction}
\end{align}
By construction, $\M_G$ is minimal and its underlying graph $H$ has a unique bottom strongly connected component, namely the self-loop in $e_0$. 
Each path $p = e_n \leadsto e_0$ in $H$ passes through all edge states, and 
the set of labels of the endpoint states in $p$ is a vertex cover of $G$. 
Since $e_0, \dots, e_n$ have pairwise distinct labels, we have that $G$ has a vertex 
cover of size at most $h$ iff there exists a path in $H$ from $e_n$ to $e_0$
that has at most $n+1+h$ different labels. Thus, \eqref{eq:SAreduction} follows by 
Lemma~\ref{lem:SAasFindingSCCinGraph}.
\end{proof}


\section{An Expectation Maximization-like Heuristic}
\label{sec:EM}
In this section we describe an approximation algorithm for computing near-optimal solutions of $\CBA$
for an arbitrary instance $\instance{\M,k}$.

Given an initial approximant $\N_0 \in \MC{}$, the algorithm produces a sequence of MCs $\N_0, \N_1, \ldots$ in $\MC{}$ having successively decreased distance from $\M$.
We defer until later a discussion of how the initial MC $\N_0$ is chosen.
The procedure is described in Algorithm~\ref{alg:update}. 

\begin{algorithm}[t]
    \algsetup{linenodelimiter=.}
    \caption{Approximate Minimization -- Expectation Maximization-like heuristic}
    \begin{algorithmic}[1]  
    \REQUIRE $\M = (M,\tau,\ell)$, $\N_0 = (N, \theta_0, \alpha)$, and $h \in \naturals$.
    \STATE $i \gets 0$
    \REPEAT 
    \STATE $i \gets i+1$
    \STATE compute $\C \in \coupling{\M}{\N_{i-1}}$ such that $\dist(\M,\N_{i-1}) = \discr{\C}(\M,\N_{i-1})$ \label{line:optC}
    \STATE $\theta_{i} \gets \textsc{UpdateTransition}(\theta_{i-1},\C)$ \label{line:update}
    \STATE $\N_i \gets (N, \theta_i, \alpha)$ \label{line:updateN}
    \UNTIL{$\dist(\M,\N_{i}) > \dist(\M,\N_{i-1})$ or $i \geq h$}
    \RETURN $\N_{i-1}$
    \end{algorithmic}
    \label{alg:update}
\end{algorithm}

The intuitive idea of the algorithm is to iteratively update the initial MC by assigning relatively greater probability to transitions that are most representative of the behavior of the MC $\M$ w.r.t.\ $\dist$. The procedure stops when the last iteration has not yield an improved approximant w.r.t.\ the preceding one. The input also includes a parameter $h \in \naturals$ that bounds the number of iterations. Furthermore, to simplify the exposition and avoid computational issues, we assume that $\dist(\M,\N_0) < 1$, $\M$ is minimal\footnote{In case $\M$ is not minimal, one can efficiently replaced it with its bisimilarity quotient~\cite{Baier96,DerisaviHS03}.}, $|M| > |N|$ and both $\tau$ and $\theta_0$ are rational transition functions.

The rest of the section explains two heuristics used in the \textsc{UpdateTransition} function invoked at line~\ref{line:update}. This function shall return the transition probabilities for the successive approximant (see line~\ref{line:updateN}). 
The two heuristics are both based upon an analysis of the coupling structure $\C \in \coupling{\M}{\N_{i-1}}$ constructed at line~\ref{line:optC}.

Define $\ddiscr{\C}$ to be the least fixed-point of the following functional operator on 
\mbox{$1$-bounded} real-valued functions $d \colon M \times N \to [0,1]$ (ordered point-wise):
\begin{equation}
    B^{\C}_{\lambda}(d)(m,n) = 
    \begin{cases}
    	1 &\text{if $\discr{\C}(m,n) = 0$} \\
    	0 &\text{if $\ell(m) \neq \alpha(n)$} \\
	(1 - \lambda) + \lambda \sum_{u \in M, v \in N} d(u,v) \cdot \C(m,n)(u,v) & \text{otherwise} \,.
    \end{cases}
    \label{eq:B}
\end{equation}
By Theorem~\ref{th:mincoupling}, the relation $R_\C = \set{(m,n)}{\discr{\C}(m,n) = 0}$ is easily shown to be
a bisimulation, specifically, the greatest bisimulation induced by $\C$. 

Let $\C_\lambda$ be the MC obtained by augmenting $\C$ with an `sink' state $\bot$ to 
which any other state moves with probability $(1-\lambda)$. Intuitively, the value $\ddiscr{\C}(m,n)$
can be interpreted as the reachability probability in $\C_\lambda$ of either hitting the sink state or a pair of bisimilar states in $R_\C$ along a path formed only by pairs of states with identical labels 
starting from $(m,n)$.

\begin{lemma} \label{lem:dualdiscrepancy}
For all $m \in M$ and $n \in N$, $\beta^\C_\lambda(m,n) = 1 - \discr{\C}(m, n)$.
\end{lemma}
\begin{proof}
We prove the equivalent statement $\discr{\C} = \mathbf{1} - \beta^\C_\lambda$. 
Consider the following operator
\begin{equation*}
    G^{\C}_{\lambda}(d)(m,n) = 
    \begin{cases}
    	0 &\text{if $\discr{\C}(m, n) = 0$} \\
    	1 &\text{if $\ell(m) \neq \ell(n)$} \\
	\lambda \sum_{u \in M, v \in N} d(u,v) \cdot \C(m,n)(u,v) & \text{otherwise} \,.
    \end{cases}
\end{equation*}
One can easily show that $G^{\C}_{\lambda}$ is monotonic, thus by Knaster-Tarski's fixed-point
theorem it admits least and greatest fixed points, say $d$ and $d'$ respectively.  
We prove that $G^{\C}_{\lambda}$ has a unique fixed showing that $d = d'$. 
We proceed by contradiction. Assume that $d \sqsubset d'$. 

Consider the set $M = \set{(m,n)}{d'(m,n) - d(m,n) = \norm{d' - d}}$. By definition of $M$ and $G^{\C}_{\lambda}$ we have that $(m,n) \in M$ implies that $\discr{\C}(m, n) \neq 0$ and $\ell(m) = \ell(n)$.
Consider $(m,n) \in M$ then we have
\begin{align*}
\norm{d' - d} &= d'(m,n) - d(m,n) \tag{def.\ $M$} \\
&= \textstyle \lambda \sum_{u \in M, v \in N} (d'(u,v) - d(u,v)) \cdot \C(m,n)(u,v)  \tag{def.\ $G^{\C}_{\lambda}$}\\
&\leq \textstyle \lambda \sum_{u \in M, v \in N} \norm{d' - d} \cdot \C(m,n)(u,v) 
	\tag{def.\ $\norm{d' - d}$} \\
&= \lambda \cdot \norm{d' - d} \tag{$\C(m,n) \in \D(M \times N)$} \,.
\end{align*}
There are two possible cases. If $\lambda < 1$ then the above inequality implies that $\norm{d' - d} = 0$ which contradicts $d \neq d'$. If $\lambda = 1$, the above inequality implies that the support of $\C(m,n)$ is included in $M$. By the generality of $m,n$, we have that $(m,n) \in M$ implies $\mathit{support}(\C(m,n)) \subseteq M$, but this implies that $\discr{\C}(m, n) = 0$, leading to a contradiction. 

By definition of $ \Gamma^{\C}_{\lambda}$, it's immediate to see that $\discr{\C} = G^{\C}_{\lambda}(\discr{\C})$. We will complete the proof by showing that also $\mathbf{1} - \beta^\C_\lambda$ is a fixed point for $G^{\C}_{\lambda}$. Let $m \in M$ and $n \in N$, if $\discr{\C}(m, n) = 0$ or $\ell(m) \neq \alpha(n)$ it trivially holds that $1 - \beta^\C_\lambda(m,n) = G^{\C}_{\lambda}(\mathbf{1} - \beta^\C_\lambda)(m,n)$. If $\discr{\C}(m, n) > 0$ and $\ell(m) = \alpha(n)$, then the following equalities hold
\begin{align*}
G^{\C}_{\lambda}(\mathbf{1} - \beta^\C_\lambda)(m,n) 
&= \textstyle \lambda \sum_{u \in M, v \in N} \big(1 - \beta^\C_\lambda(u,v) \big) \C(m,n)(u,v) \tag{by def.\ $G^{\C}_{\lambda}$} \\
&= \textstyle \lambda - \lambda \sum_{u \in M, v \in N} \beta^\C_\lambda(u,v) \cdot \C(m,n)(u,v) \tag{$\C(m,n) \in \D(M \times N)$} \\
& = 1 - \beta^\C_\lambda(m,n) \tag{$\beta^\C_\lambda(m,n) = B^{\C}_{\lambda}(\beta^\C_\lambda)(m,n)$}
\end{align*} \qed
\end{proof}

\noindent From equation \eqref{eq:CBAdiscr} and Lemma~\ref{lem:dualdiscrepancy}, we can turn the problem $\CBA$ as
\begin{equation}
  \argmax \set{ \beta^\C_\lambda(\M, \N) }{ \N \in \MC{L(\M)},\, \C \in \coupling{\M}{\N} } \,.
  \label{eq:dualCBA}
\end{equation} 
Equation~\eqref{eq:dualCBA} says that a solution of $\CBA$ is the right marginal of a coupling structure $\C$
such that $\C_\lambda$ maximizes the probability of generating paths with prefix in ${\cong}^*(R_\C \cup \bot)$ starting from the pair $(m_0,n_0)$ of initial states\footnote{We borrowed notation from regular expressions, such as language union, concatenation, and Kleene star, to express the set of finite paths ${\cong}^*(R_\C \cup \bot)$ as a language over the alphabet $(M \times N) \cup \bot$.}, where ${\cong} = \set{(m,n) \notin R_\C }{\ell(m) = \alpha(n)}$.

In the rest of the section we assume $\N_{i-1} \in \MC{}$ to be the current approximant with associated 
coupling structure $\C \in \coupling{\M}{\N_{i-1}}$ as in line 4 in Algorithm~\ref{alg:update}.

\subparagraph*{The ``Averaged Marginal'' Heuristic.}
The first heuristic is inspired by the Expectation Maximization (EM) algorithm described in~\cite{BenediktLW13}.
The idea is to count the expected number of occurrences of the transitions in $\C$ in the set of paths ${\cong}^*R_\C$ and, in accordance with~\eqref{eq:dualCBA}, updating $\C$ by increasing the probability of the transitions that were contributing the most. 

For each $m,u \in M$ and $n,v \in N$ let $Z^{m,n}_{u,v} \colon ((M \times N)\cup \bot)^\omega \to \naturals$ be the random variable that counts the number of occurrences of the edge $((m,n)(u,v))$ in a prefix in ${\cong}^*(R_\C \cup \bot)$ of the given path.
We denote by $\E[\C]{Z^{m,n}_{u,v}}$ the expected value of $Z^{m,n}_{u,v}$ with respect to the probability distribution induced by $\C_\lambda$. Using these values we define the optimization problem $\text{EM}\instance{\N,\C}$:
\begin{align}
\text{maximize}\;\;& \textstyle\sum_{m,u \in M} \sum_{n,v \in N} \E[\C]{Z^{m,n}_{u,v}} \cdot \ln (c^{m,n}_{u,v}) & \notag \\
\text{such that}\;\;
&\textstyle\sum_{v \in N} c^{m,n}_{u,v} = \tau(m)(u)  && \text{$m,u \in M$, $n \in N$}  \label{eq:droplater} \\ 
&\textstyle\sum_{u \in M} c^{m,n}_{u,v} = \theta_{n,v} && \text{$m \in M$, $n,v \in N$} \label{eq:dropmarginal} \\
& c^{m,n}_{u,v} \geq 0 &&\text{$m,u \in M$, $n,v \in N$} \notag
\end{align}
%
A solution of $\text{EM}\instance{\N,\C}$ describes a Markov chain $\N' = (N, \theta', \alpha)$ and a coupling structure $\C \in \coupling{\M}{\N'}$ where, for arbitrary $m \in M$ and $n \in N$
\begin{align*}
&\theta'(n) = \textstyle\sum_{v \in N} \theta_{n,v} \cdot 1_v \,, &&
\C'(m,n) = \textstyle\sum_{(u,v) \in M \times N} c^{m,n}_{u,v} \cdot 1_{(u,v)} \,.
\end{align*} 
The above can be used to improve a pair $\instance{\N,\C}$ in the sense of~\eqref{eq:dualCBA}. 
\begin{theorem}\label{thm:EM}
If $\ddiscr{\C}(\M,\N) > 0$, then an optimal solution for $\text{EM}\instance{\N,\C}$ describes a Markov chain $\N' \in \MC{}$ and a coupling structure $\C' \in \coupling{\M}{\N'}$ satisfying the inequality $\ddiscr{\C'}(\M,\N') \geq \ddiscr{\C}(\M,\N)$.
\end{theorem}
\begin{proof}
For any measurable set $A \subseteq (M \times N)^\omega$, and $\C \in \coupling{\M}{\N}$, we denote by $\Pr[{\C_\lambda}]{A}$ the probability that a run of the chain $\C_\lambda$ belongs to $A$. To shorten the notation, for $B \subseteq (M \times N)^*$ (resp.\ $\pi \in (M \times N)^*$) we write $\Pr[{\C}]{B}$ (resp.\ $\Pr[{\C}]{\pi}$) to indicate $\Pr[{\C_\lambda}]{B(M \times N)^\omega}$ (resp.\ $\Pr[{\C_\lambda}]{\pi(M \times N)^\omega}$).

Recall that, $\ddiscr{\C}(m_0,n_0) = \Pr[{\C}]{G}$ where $G = ({\cong}^*(R_\C \cup \bot))$ that is the probability that $\C_\lambda$ generates a path with prefix in ${\cong}^*R_\C$ or ${\cong}^*\bot$ starting from $(m_0,n_0)$. 

Consider the following inequalities 
\begin{align*}
&\ln \ddiscr{\C'}(m_0,n_0) - \ln \ddiscr{\C}(m_0,n_0) \\
&= \ln \frac{\Pr[{\C'}]{G}}{\Pr[{\C}]{G}} \\
&= \ln \frac{\sum_{\pi \in G}\Pr[\C']{G \mid \pi} \cdot \Pr[\C']{\pi} }{\Pr[\C]{G}} \\
&= \ln \sum_{\pi \in G} \frac{\Pr[\C']{G \mid \pi} \cdot \Pr[\C']{\pi}}{\Pr[\C]{G}} 
\cdot \frac{\Pr[\C]{\pi \mid G}}{\Pr[\C]{\pi \mid G}} \\
&\geq \sum_{\pi \in G} \Pr[\C]{\pi \mid G} \cdot \ln \frac{\Pr[\C']{G \mid \pi} \cdot \Pr[\C']{\pi}}{\Pr[\C]{G} \cdot \Pr[\C]{\pi \mid G}} 
\tag{by Jensen's inequality} \\
&= \frac{1}{\ddiscr{\C}(m_0,n_0)} \sum_{\pi \in G} \Pr[\C]{\pi} \cdot \ln \frac{\Pr[\C']{\pi}}{\Pr[\C]{\pi}} \\
&= \frac{1}{\ddiscr{\C}(m_0,n_0)} (Q' - Q)
\end{align*}
where $Q' = \sum_{\pi \in G} \Pr[\C]{\pi} \cdot \ln \Pr[\C']{\pi}$ and $Q = \sum_{\pi \in G} \Pr[\C]{\pi} \cdot \ln \Pr[\C]{\pi}$. Rearranging we have 
\begin{equation}
\ddiscr{\C}(m_0,n_0) \cdot \big(\ln \ddiscr{\C'}(m_0,n_0) - \ln \ddiscr{\C}(m_0,n_0) \big) \geq Q' - Q \,.
\label{eq:rearrange}
\end{equation}
The logarithm is an strictly increasing function, we have that 
\begin{equation*}
\ln \ddiscr{\C'}(m_0,n_0) - \ln \ddiscr{\C}(m_0,n_0) \geq 0 
\qquad \text{iff} \qquad  \ddiscr{\C'}(m_0,n_0) \geq \ddiscr{\C}(m_0,n_0) \,,
\end{equation*}
 therefore by \eqref{eq:rearrange} we conclude that $Q' \geq Q$ implies $\ddiscr{\C'}(m_0,n_0) \geq \ddiscr{\C}(m_0,n_0)$.

Thus, inequality~\ref{eq:rearrange} suggests that the best choice of $\C'$ is that which maximizes
$Q'$ as a function of $\C'$. Expanding the definition of $Q'$ we obtain
\begin{equation*}
Q' = \sum_{\pi \in G} \Pr[\C]{\pi} \cdot \ln \Pr[\C']{\pi} 
= \sum_{\pi \in G} \Pr[\C]{\pi} \Big( \ln \iota(\pi_0) + \sum_{i = 0}^{|\pi|-1} \ln{\C'_\lambda(\pi_{i})(\pi_{i+1})}\Big) \, ,
\end{equation*}
where $\iota$ is defined as $\iota(x) = 1$ if $x = (m_0,n_0)$; $0$ otherwise. 

Recall that $Z^{m,n}_{u,v}$ is the random variable that counts the number of occurrences of the edge $((m,n)(u,v))$ in a prefix in $G$ for a given path. Then $Q'$ can be rewritten as 
\begin{equation*}
\sum_{\pi \in G} \Pr[\C]{\pi} \Big( \ln(\iota(\pi_0)) + \sum_{m,u \in M} \sum_{n,v \in N} Z^{m,n}_{u,v}(\pi) \ln \C'(m,n)(u,v) \Big) \, . 
\end{equation*}
Therefore the coupling structure $\C'$ that maximizes the above is obtained as
\begin{align*}
&\argmax_{\C'} \sum_{\pi \in G} \Pr[\C]{\pi} \sum_{m,u \in M} \sum_{n,v \in N} Z^{m,n}_{u,v}(\pi) \ln \C'({m,n})({u,v})  \tag{eliminating constants} \\
&\argmax_{\C'} \sum_{m,u \in M} \sum_{n,v \in N}  \E{Z^{m,n}_{u,v} \mid \C} \cdot \ln \C'({m,n})({u,v})
\end{align*}
Since $\C'$ has to range among coupling structures of the form $\C' \in \coupling{\M}{\N'}$ for some chain $\N'$ with the same states as $\N$ we conclude that an optimal solution of $\text{EM}\instance{\N,\C}$ describes a coupling $\C'$ such that $Q' \geq Q$.
As above said, this implies $\ddiscr{\C'}(\M,\N') \geq \ddiscr{\C}(\M,\N)$. \qed
\end{proof}

Unfortunately, $\text{EM}\instance{\N,\C}$ does not have an easy analytic 
solution and turns out to be inefficiently solved by nonlinear optimization methods. 
In contrast, the relaxed optimization problem obtained by dropping the constraints~\eqref{eq:dropmarginal} 
has a simple analytic solution\footnote{By abusing the notation, whenever the nominator is $0$, we consider the entire expression equal to $0$, regardless of any division by $0$. The same convention is used implicitly in the rest of the section.}:
\begin{equation*}
c^{m,n}_{u,v} = \frac{\tau(m)(n) \cdot \E[\C]{Z^{m,n}_{u,v}}}{\sum_{x \in N} \E[\C]{Z^{m,n}_{u,x}}} \,,
\end{equation*}
and the first heuristic at line~\ref{line:update}, updates $\theta_i$ as follows
\begin{equation*}
  \theta_i(n)(v) =
\begin{cases}
  \theta_{i-1}(n)(v) & \text{if $\exists m \in M.\, n \mathrel{R_\C} m$} \\
  \displaystyle
  \frac{\sum_{m,u \in M} c^{m,n}_{u,v}}{\sum_{x \in N} \sum_{m,u \in M} c^{m,n}_{u,x} }
  &\text{otherwise.}
\end{cases}
\end{equation*}
Recall that the $c^{m,n}_{u,v}$ above may not describe a coupling structure because of the dropping of the constraints~\eqref{eq:dropmarginal} . Nevertheless we recover the transition probability $\theta_i$, from it as the average of the right marginals.

\subparagraph*{The ``Averaged Expectations'' Heuristic.}
In contrast to the previous case, the second heuristic will update $\theta_i$ by directly averaging the expected values of $Z^{m,n}_{u,v}$ as follows  
\begin{equation*}
  \theta_i(n)(v) = \left\{
\begin{aligned}
  &\theta_{i-1}(n)(v) && \text{if $\exists m \in M.\, n \mathrel{R_\C} m$} \\
  &\frac{\sum_{m,u \in M} \E[\C]{Z^{m,n}_{u,v}}}{\sum_{x \in N} \sum_{m,u \in M} \E[\C]{Z^{m,n}_{u,x}}} 
  && \text{otherwise}\,.
\end{aligned} \right.
\end{equation*}

\subparagraph*{Computing the Expected Values.}

We compute $\E[\C]{Z^{m,n}_{u,v}}$ using a variant of the \emph{forward-backward} algorithm for hidden Markov models. 
Let $Z^{m,n} \colon ((M \times N) \cup \bot)^\omega \to \naturals$ be the random variable that counts the number of occurrences of the pair $(m,n)$ in a prefix in ${\cong}^*(R_\C \cup \bot)$ of the path. 
We compute the expected value of $Z^{m,n}$ w.r.t.\ the probability induced by $\C_\lambda$ as the solution $z_{m,n}$ of the following system of equations
\begin{equation}
z_{m,n} = 
\begin{cases}
0 & \text{if $m \not\cong n$} \\
\iota(m,n) + \lambda \sum_{u,v} (z_{u,v} + 1) \cdot \C(u,v)(m,n) &\text{otherwise} \,,
\end{cases}
\label{eq:zeta}
\end{equation}
where $\iota$ is defined as $\iota(x) = 1$ if $x = (m_0,n_0)$; $0$ otherwise.
Then, the expected value of $Z^{m,n}_{u,v}$ with respect to the probability distribution induced by $\C_\lambda$ is given by 
\begin{equation}
\E[\C]{Z^{m,n}_{u,v}} = \lambda \cdot z_{m,n} \cdot \C(m,n)(u,v) \cdot \ddiscr{\C}(u,v) \,.
\label{eq:EZ}
\end{equation}

\subparagraph*{Complexity of the Heuristics.} 
It is worth noting that Algorithm~\ref{alg:update} runs in polynomial time in the size of its input. 
Computing the coupling structure $\C$ in line~\ref{line:optC}, can be performed in polynomial time in the size of $\M \oplus \N_0$~\cite{ChenBW12,BacciLM:ictac15}\footnote{As pointed out in~\cite{BacciLM:tacas13} the method proposed in~\cite{ChenBW12} turns out to be slow in practice. In our implementation of Algorithm~\ref{alg:update} the computation of the coupling structure $\C$ in line~\ref{line:optC} is performed by using the on-the-fly method proposed in~\cite{TangB17}}.

The update step requires one to compute $\E[\C]{Z^{m,n}_{u,v}}$ for each $m,u \in M$ and $n,v \in N$. By Equation~\eqref{eq:EZ}, this can be done by solving two systems of linear equations, namely \eqref{eq:zeta} and \eqref{eq:B}, each with $O(|M| \cdot |N|)$ unknowns. 

\subparagraph*{Choosing the Initial Approximant.} Similarly to EM algorithms, the choice of the initial approximant $\N_0$ may have a significant effect on the quality of the solution. 

Notice that the requirement $\dist(\M,\N_0) < 1$ is fundamental for both the heuristics, otherwise $\E[\C]{Z^{m,n}_{u,v}} = 0$ for all $m,u \in M$ and $n,v \in N$. In case $\N_0$ does not satisfy the above condition, one should replace it with a better approximant. We have seen that this is trivial to do when $\lambda < 1$, but for $\lambda = 1$ it may be hard to find one (\cf\ Theorem~\ref{th:NP-completeSA}), in which case one may need to increase the size of the initial approximant. 

For the labeling of the states, one should follow Lemma~\ref{lem:usefullabels}. As for the choice of the underlying structure one shall be guided by Lemma~\ref{lem:SAasFindingSCCinGraph}.  
However, due to Theorem~\ref{th:NPhardAM}, it seems unlikely to have generic good strategies for selecting a starting approximant candidate. Nevertheless, good selections for the transition probabilities and the number of states may be suggested by looking at the problem instance.

Finally, by the assumption made on the rationality of the transition functions of $\M$ and $\N_0$, Algorithm~\ref{alg:update} will always return a chain with rational transition function. In light of Example~\ref{ex:minimalirrational}, this means that there are cases where Algorithm~\ref{alg:update} will never return an (exact) optimal solution, regardless from the choice of (the rational) initial approximant. Nevertheless, we will see that the we are still able to provide good sub-optimal solutions. 

\begin{example}
Consider the Markov chain $\M$ from Example~\ref{ex:minimalirrational}. We show some iterations of Algorithm~\ref{alg:update} comparing averaged marginals (AM) and the averaged expectation (AE) heuristics. The tests have been performed staring from two instances of the parametric MC $\N_{x,y}$ described in Example~\ref{ex:minimalirrational}. 
\begin{table}[h]
\centering
\setlength{\tabcolsep}{0.5em}
\renewcommand{\arraystretch}{1.3}
\begin{tabular}{|c|c|c|c|c|}
\hline
\multicolumn{5}{|c|}{Test $1$} \\
\hline
Heur. & iter & $x$ & $y$ & $\dist[1](\M, \N_{x,y})$ \\ \hline\hline
\multirow{3}{*}{AM}&
0 & $\frac{1}{10}$ & $\frac{1}{2}$ & $\frac{7629}{10000} \approx 0.76$ \\ \cline{2-5}
& $\mathbf{1}$ & $\mathbf{\frac{79}{150}}$ & $\mathbf{\frac{21}{100}}$ & $\mathbf{\frac{231233}{421875} \approx 0.55}$ \\  \cline{2-5}
& 2 & $\frac{79}{150}$ & $\frac{21}{100}$ & $\frac{231233}{421875} \approx 0.55$ \\ 
\hline
\multirow{2}{*}{AE}&
$\mathbf{0}$ & $\mathbf{\frac{1}{10}}$ & $\mathbf{\frac{1}{2}}$ & $\mathbf{\frac{7629}{10000} \approx 0.762}$ \\ \cline{2-5}
& 1 & $\frac{131}{1215}$ & $\frac{308}{405}$ & $\frac{231233}{421875} \approx 0.763$ \\ \hline
\end{tabular}
\\[1ex]
\begin{tabular}{|c|c|c|c|c|}
\hline
\multicolumn{5}{|c|}{Test $2$} \\
\hline
Heur. & iter & $x$ & $y$ & $\dist[1](\M, \N_{x,y})$ \\ \hline\hline
\multirow{3}{*}{AM}&
0 & $\frac{1}{10}$ & $\frac{1}{8}$ & $\frac{1707}{2000} \approx 0.85$ \\  \cline{2-5}
& $\mathbf{1}$ & $\mathbf{\frac{79}{150}}$ & $\mathbf{\frac{21}{100}}$ & $\mathbf{\frac{231233}{421875} \approx 0.55}$ \\  \cline{2-5}
& 2 & $\frac{79}{150}$ & $\frac{21}{100}$ & $\frac{231233}{421875} \approx 0.55$ \\ 
\hline
\multirow{3}{*}{AE}&
0 & $\frac{1}{10}$ & $\frac{1}{8}$ & $\frac{1707}{2000} \approx 0.85$ \\  \cline{2-5}
& $\mathbf{1}$ & $\mathbf{\frac{79}{595}}$ & $\mathbf{\frac{66}{119}}$ & $\mathbf{\frac{3171910751}{4212897500} \approx 0.752}$ \\ \cline{2-5}
& 2 & $\frac{38983972}{287248825}$ & $\frac{211382913}{287248825}$ & $\frac{71603676151221064844103643}{94805770890127757886062500} \approx 0.755$ \\ \hline
\end{tabular}
\caption{Tests on the accuracy of Algorithm~\ref{alg:update} on the Markov chain $\M$ from Example~\ref{ex:minimalirrational}.}
\label{table:irrational}
\end{table}

Table~\ref{table:irrational} shows how the two heuristics update the initial approximant and how the distance between $\M$ and the current approximant evolves at each iteration. Notably, the averaged marginals heuristics sets the parameter $y$ to $\frac{21}{100}$ in both tests, reaching the its theoretical optimal value. Regarding the parameter $x$ the heuristics sets its value to $\frac{79}{150}$ reaching an absolute error of $0.23$ from its (irrational) theoretical optimal value. 
As Table~\ref{table:irrational} shows, in this example, the averaged expectation heuristic suffers from its oversimplified update procedure. 
It is worth noting how the quality of the outcome may not be influenced by the quality of the initial estimate. \qed
\end{example}

\section{Experimental Results}
\label{sec:experiments}
We evaluate the performances of Algorithm~\ref{alg:update} by comparing the two proposed heuristics on 
two classical case studies: the \emph{IPv4 zeroconf protocol} from~\cite[Ex.10.5]{BaierK08} (\cf\ Figure~\ref{IPv4}) and the \emph{drunkard's walk} (\cf\ Figure~\ref{DrkW}).

\begin{figure}[h]
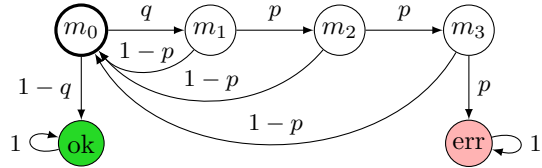

\def\h{1.7cm}
\def\v{1.5cm}
\centering
\tikz[labels, baseline= {(m0.center)}, 
n1/.style={}, n2/.style={fill=red!30}, n3/.style={fill=green!70!gray}]{ 
  \draw (0,0) node[label, initial, n1] (m0) {$m_0$}; 
  \draw ($(m0)+(right:\h)$) node[label, n1] (m1) {$m_1$};
  \draw ($(m1)+(right:\h)$) node[label, n1] (m2) {$m_2$};
   \draw ($(m2)+(right:\h)$) node[label, n1] (m3) {$m_3$};
  \draw ($(m3)+(down:\v)$) node[label, n2] (err) {err};
  \draw ($(m0)+(down:\v)$) node[label, n3] (ok) {ok};
  \path[-latex, font=\small]
    (m0) edge node[above] {$q$} (m1)
    	    edge node[left] {$1-q$} (ok)
    (m1) edge node[above] {$p$} (m2)
            edge[bend left=50] node[above] {$1-p$} (m0)
    (m2) edge node[above] {$p$} (m3)
            edge[bend left=50] node[above] {$1-p$} (m0)
    (m3) edge node[right] {$p$} (err)
            edge[bend left=60] node[above] {$1-p$} (m0)
    (err) edge[loop right] node[right] {$1$} (err)
    (ok) edge[loop left] node[left] {$1$} (ok)
    ;
}
\caption{Markov chain $\text{IPv4}(p,q,n)$ of the IPv4 zeroconf protocol for $n = 3$ probes. The parameter $p \in [0,1]$ models the probability that no reply is received and the next probe is sent, while $q \in [0,1]$ models the probability that the randomly chosen IP address is already used.}
\label{IPv4}
\end{figure}
\begin{figure}[h]
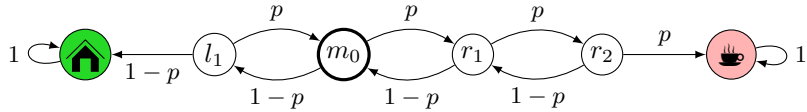

\def\h{1.7cm}
\def\v{1.5cm}
\def\house{\hbox{\kern3pt \vbox to13pt{}%
   \pdfliteral{q 0 0 m 0 5 l 5 10 l 10 5 l 10 0 l 7 0 l 7 5 l 3 5 l 3 0 l f
               1 j 1 J -2 5 m 5 12 l 12 5 l S Q }%
   \kern 13pt}}
\centering
\tikz[labels, baseline= {(m0.center)}, 
n1/.style={}, n2/.style={fill=red!30}, n3/.style={fill=green!70!gray}]{ 
  \draw (0,0) node[label, initial, n1] (m0) {$m_0$}; 
  \draw ($(m0)+(right:\h)$) node[label, n1] (r1) {$r_1$};
  \draw ($(r1)+(right:\h)$) node[label, n1] (r2) {$r_2$};
  \draw ($(r2)+(right:\h)$) node[label, n2,font={\large}] (co) {\Coffeecup};
  \draw ($(m0)+(left:\h)$) node[label, n1] (l1) {$l_1$};
  \draw ($(l1)+(left:\h)$) node[label,n3,inner sep=-.5pt] (home) {\house};
  \path[-latex, font=\small]
    (m0) edge[bend left] node[above] {$p$} (r1)
    	    edge[bend left] node[below] {$1-p$} (l1)
    (r1) edge[bend left] node[above] {$p$} (r2)
    	    edge[bend left] node[below] {$1-p$} (m0)
    (r2) edge node[above] {$p$} (co)
    	    edge[bend left] node[below] {$1-p$} (r1)
    (l1) edge[bend left] node[above] {$p$} (m0)
    	    edge node[below] {$1-p$} (home)
    (co) edge[loop right] node[right] {$1$} (co)
    (home) edge[loop left] node[left] {$1$} (home)
    ;
}
\caption{Markov chain $\text{DrkW}(p,n, k)$ of the drunkard's walk for $n = 1$ steps to home and $k = 2$ steps to the bar. The parameter $p \in [0,1]$ models the probability of moving one step toward the bar.}
\label{DrkW}
\end{figure}

Table~\ref{tab:EMexperiments} shows the results of our tests\footnote{
The tests have been made using a prototype implementation coded
in $\text{Mathematica}^{\circledR}$ (available at \url{people.cs.aau.dk/giovbacci/tools.html})
running on an Intel Core-i5 2.5 GHz processor with 8GB of DDR3 RAM 1600MHz.}.
The experiments have been performed by running our algorithm on a number of instances $\instance{\M,k}$
of increasing size, where $\M$ is an instance of either the IPv4 protocol or the drunkard's walk. 
 
As initial approximant we use a suitably small instance of the same model, with parameters $p$ and $q$ chosen randomly in the open interval $(0,1)$.
For each experiment we report the discount factor $\lambda$; the distance to the original model respectively
from $\N_0$ and $\N_h$, where $h$ is the total number of iterations; and execution time (in seconds).
We compare the two heuristics, averaged marginals (AM) and averaged expectation (AE), on the same initial approximant.

The results obtained on the IPv4 protocol show significant improvements between the initial
and the returned approximant. Notably, these are obtained in very few iterations of the update procedure.
On this model, AM gives approximants of better quality compared with those obtained using AE; however
AE seems to be slightly faster than AM. Both the heuristics can handle instances of size up to ${\sim}100$ states.
On the drunkard's walk model, the two heuristics exhibit opposite results w.r.t.\ the previous experiment: AE provides the best solutions with fewer iterations and significantly lower execution times.

\begin{table}[ht]
\begin{center}
\setlength{\tabcolsep}{1.5ex}
\begin{tabular}[t]{|c|c|c|c|c|c|c|c|c|c|c|}
\hline
 \multirow{2}{*}{Case} & 
 \multirow{2}{*}{$|M|$} &
 \multirow{2}{*}{$k$} &
 \multicolumn{4}{c|}{$\lambda = 1$} &
 \multicolumn{4}{c|}{$\lambda = 0.8$}
 \\ \cline{4-11}
 &&& $\dist$-init & $\dist$-final & \# & time
 & $\dist$-init & $\dist$-final & \# & time \\
 \hline\hline
 \multirow{6}{*}{$\begin{array}{c}\text{IPv4} \\ \text{(AM)} \end{array}$} 
 & 23 & 5 & 0.775 & 0.054 & 3 & 4.8 		& 0.576 & 0.025 & 3 & 4.8 \\
 & 53 & 5 & 0.856 & 0.062 & 3 & 25.7	& 0.667 & 0.029 & 3 & 25.9 \\
 & 103 & 5 & 0.923 & 0.067 & 3 & 116.3	& 0.734 & 0.035 & 3 & 116.5 \\
 \cline{2-11}
 & 53 & 6 & 0.757 & 0.030 & 3 & 39.4	& 0.544 & 0.011 & 3 & 39.4 \\
 & 103 & 6 & 0.837 & 0.032 & 3 & 183.7	& 0.624 & 0.017 & 3 & 182.7 \\
 & 203 & 6 & -- & -- & -- & TO			& -- & -- & -- & TO \\
 \hline
 \multirow{6}{*}{$\begin{array}{c}\text{IPv4} \\ \text{(AE)} \end{array}$} 
 & 23 & 5 & 0.775 & 0.109 & 2 & 2.7		& 0.576 & 0.049 & 3 & 4.2 \\
 & 53 & 5 & 0.856 & 0.110 & 2 & 14.2	& 0.667 & 0.049 & 3 & 21.8 \\
 & 103 & 5 & 0.923 & 0.110 & 2 & 67.1	& 0.734 & 0.049 & 3 & 100.4 \\
 \cline{2-11}
 & 53 & 6 & 0.757 & 0.072 & 2 & 21.8	& 0.544 & 0.019 & 3 & 33.0 \\
 & 103 & 6 & 0.837 & 0.072 & 2 & 105.9	& 0.624 & 0.019 & 3 & 159.5 \\
 & 203 & 6 & -- & -- & -- & TO			& -- & -- & -- & TO \\
 \hline\hline
 \multirow{3}{*}{$\begin{array}{c}\text{DrkW} \\ \text{(AM)} \end{array}$} 
 & 39 & 7 & 0.565 & 0.466 & 14 & 259.3 	& 0.432 & 0.323 & 14 & 252.8 \\
 & 49 & 7 & 0.568 & 0.460 & 14 & 453.7 	& 0.433 & 0.322 & 14 & 420.5 \\
 & 59 & 8 & 0.646 & -- & -- & TO			& 0.423 & -- & -- & TO \\
 \hline
 \multirow{3}{*}{$\begin{array}{c}\text{DrkW} \\ \text{(AE)} \end{array}$} 
 & 39 & 7 & 0.565 & 0.435 & 11 & 156.6	& 0.432 & 0.321 & 2 & 28.6 \\
 & 49 & 7 & 0.568 & 0.434 & 10 & 247.7 	& 0.433 & 0.316 & 2 & 46.2 \\
 & 59 & 8 & 0.646 & 0.435 & 10 & 588.9	& 0.423 & 0.309 & 2 & 115.7 \\
 \hline
\end{tabular}
\end{center}

\caption{Comparison of the performance of Algorithm~\ref{alg:update} on the IPv4 zeroconf protocol
and the classic Drunkard's Walk w.r.t.\ the heuristics AM and AE.
}
\label{tab:EMexperiments}
\end{table}

\section{Conclusions and Future Work}
\label{sec:conclusion}
To the best of our knowledge, this is the first paper addressing the complexity of the \emph{optimal} approximate minimization of MCs w.r.t.\ a behavioral metric semantics. Even though for 
a good evaluation of our heuristics more tests are needed, the current results seem promising.
Moreover, in the light of~\cite{ChenBW12,BacciLM:ictac15}, relating the probabilistic bisimilarity distance to the LTL-model checking problem as $\dist[1](\M,\N) \geq |\Pr[\M]{\varphi} - \Pr[\N]{\varphi}|$, for all $\varphi \in \text{LTL}$,
our results might be used to lead saving in the overall model checking time. A deeper study of this topic will be
the focus of future work. 

We close with an interesting open problem. Membership of $\BA$ in NP is left open. However, 
by arguments analogous to~\cite{ChenK14,EtessamiY09} and leveraging on the ideas that 
made us produce the MC in Example~\ref{ex:minimalirrational}, we suspect that 
$\BA$ is hard for the square-root-sum problem. The latter is known to be NP-hard and in PSPACE,
but membership in NP has been open since 1976. Allender et al.~\cite{Allender09} showed that it can be
decided in the 4th level of the counting hierarchy, thus it is unlikely its PSPACE-completeness.

Furthermore, in light of the relation of the bisimilarity distance $\dist$ with the LTL model checking problem (\cf\ Corollary~\ref{cor:disctvupperbound}) and the PosSLP-hardness of the model checking problem for Interval Markov chains discovered in~\cite[Theorem 3]{BenediktLW13} one may also consider to study the $\BA$ problem as a good candidate for $\text{NP}_{\mathbb{R}}$ completeness in the Blum-Shub-Smale model of computation over the real filed with order $(\mathbb{R}, <)$~\cite{BlumCSS97}.

\paragraph{Acknowledgements}
We thank the anonymous reviewers for their comments on an earlier version of the article.
We would also like to show our gratitude to Nathana\"el Fijalkow who pointed out a flaw on an earlier version of the proof of Theorem~\ref{th:NPhardAM}.

\section*{References}

\bibliography{biblio}
\end{document}